\makeatletter
\def\input@path{{template/}}
\makeatother

\documentclass[runningheads]{llncs}

\newif\ifpublish
\publishtrue

\usepackage{silence}
\usepackage{type1cm}   
\usepackage{microtype}
\usepackage{amsmath}
\usepackage{amssymb}
\usepackage{algorithm}
\usepackage[noend]{algpseudocode}
\usepackage{graphicx}
\usepackage{booktabs}
\usepackage{array}
\usepackage{tabularx}     
\usepackage{tikz}
\usetikzlibrary{arrows.meta}
\usepackage{xcolor}
\usepackage{colortbl}
\usepackage{xspace}
\usepackage{float}
\usepackage{etoolbox}
\usepackage[hypertexnames=false]{hyperref}
\usepackage{cleveref}
\usepackage{enumitem}

\setlist{nosep, leftmargin=10pt, labelsep=5pt, rightmargin=0pt}
\newcounter{clause}
\renewcommand{\theclause}{\textbf{A\arabic{clause}}}
\newcommand{\clauseitem}[1]{\item\refstepcounter{clause}\label{#1}\theclause:}
\crefname{clause}{clause}{clauses}
\Crefname{clause}{Clause}{Clauses}

\hypersetup{
    colorlinks=true,
    linkcolor={magenta!80!black},
    citecolor={green!55!black},
    urlcolor={black!55},
    pdfborder={0 0 0},
}
\crefname{section}{Section}{Sections}
\Crefname{section}{Section}{Sections}
\crefname{subsection}{Section}{Sections}
\Crefname{subsection}{Section}{Sections}
\crefname{appendix}{Appendix}{Appendices}
\Crefname{appendix}{Appendix}{Appendices}
\crefname{figure}{Fig.}{Figs.}
\Crefname{figure}{Fig.}{Figs.}
\crefname{table}{Table}{Tables}
\Crefname{table}{Table}{Tables}
\crefname{equation}{Equation}{Equations}
\Crefname{equation}{Equation}{Equations}
\crefname{algorithm}{Algorithm}{Algorithms}
\Crefname{algorithm}{Algorithm}{Algorithms}
\crefname{theorem}{Theorem}{Theorems}
\Crefname{theorem}{Theorem}{Theorems}
\crefname{lemma}{Lemma}{Lemmas}
\Crefname{lemma}{Lemma}{Lemmas}
\crefname{proposition}{Proposition}{Propositions}
\Crefname{proposition}{Proposition}{Propositions}
\crefname{corollary}{Corollary}{Corollaries}
\Crefname{corollary}{Corollary}{Corollaries}
\crefname{definition}{Definition}{Definitions}
\Crefname{definition}{Definition}{Definitions}
\crefname{remark}{Remark}{Remarks}
\Crefname{remark}{Remark}{Remarks}
\crefname{observation}{Observation}{Observations}
\Crefname{observation}{Observation}{Observations}
\crefname{myclaim}{Claim}{Claims}
\Crefname{myclaim}{Claim}{Claims}
\crefname{fact}{Fact}{Facts}
\Crefname{fact}{Fact}{Facts}
\makeatletter
\crefname{evalclaim}{Claim}{Claims}
\Crefname{evalclaim}{Claim}{Claims}
\crefname{ALG@line}{line}{lines}
\Crefname{ALG@line}{Line}{Lines}
\makeatother

\spnewtheorem{observation}{Observation}{\bfseries}{\itshape}
\spnewtheorem{myclaim}{Claim}{\bfseries}{\itshape}
\spnewtheorem{fact}{Fact}{\bfseries}{\itshape}

\newcommand{\sysname}{{\upshape\textsf{Steelhead}}\xspace}
\newcommand{\mysticeti}{Mysticeti\xspace}
\newcommand{\mahimahi}{Mahi-Mahi\xspace}
\newcommand{\codelink}{
    \ifpublish
        \href{https://github.com/asonnino/mysticeti/tree/1ae2aa7ecf13990335fc04a600ca0dcf3e763a8f}{\nolinkurl{github.com/asonnino/mysticeti}} (commit \texttt{1ae2aa7})
    \else
        \url{https://anonymous.4open.science/r/steelhead}
    \fi
}
\newcommand{\leanlink}{
    \ifpublish
        \url{https://github.com/gdanezis/lean-dag} (commit \texttt{4228bc5})
    \else
        \url{https://anonymous.4open.science/r/lean-dag}
    \fi
}

\newcommand{\para}[1]{\par\medskip\noindent\textbf{#1.}~}

\newcounter{evalclaim}
\renewcommand{\theevalclaim}{\textbf{C\arabic{evalclaim}}}
\newcommand{\claimitem}[1]{\item\refstepcounter{evalclaim}\label{#1}\theevalclaim:}

\makeatletter
\patchcmd{\@makecaption}{\small}{\scriptsize}{}{\ClassError{macros}{caption patch failed}{}}
\makeatother
\makeatletter
\patchcmd{\table}{\setlength\belowcaptionskip{10\p@}}{\setlength\belowcaptionskip{4\p@}}{}{\ClassError{macros}{table caption patch failed}{}}
\makeatother
\newcounter{subfig}[figure]

\crefname{subfig}{Fig.}{Figs.}
\Crefname{subfig}{Fig.}{Figs.}
\newcommand{\subcap}[1]{\refstepcounter{subfig}%
    \par\vspace{2pt}{\scriptsize\leftskip=0pt\rightskip=0pt\parfillskip=0pt plus 1fil
        \noindent\textbf{(\alph{subfig})}~#1\par}\vspace{3pt}}

\newcommand{\algsize}{\fontsize{6.5}{7.5}\selectfont}

\AtBeginEnvironment{algorithm}{\algsize}
\AtBeginEnvironment{algorithmic}{\algsize}   
\newcommand{\algcommentsize}{\fontsize{5.5}{6.5}\selectfont}
\algrenewcommand{\algorithmiccomment}[1]{\hfill{\algcommentsize\color{black!60}$\triangleright$ #1}}
\algrenewcommand{\alglinenumber}[1]{{\color{black!60}#1:}}
\floatstyle{boxed}
\restylefloat{algorithm}
\makeatletter
\renewcommand\floatc@plain[2]{\setbox\@tempboxa\hbox{\scriptsize{\@fs@cfont #1:} #2}%
    \ifdim\wd\@tempboxa>\hsize {\scriptsize{\@fs@cfont #1:} #2\par}%
    \else\hbox to\hsize{\hfil\box\@tempboxa\hfil}\fi}
\makeatother

\newcommand{\wl}{\ensuremath{w}\xspace}                 
\newcommand{\wlof}[1]{\ensuremath{w(#1)}}               
\newcommand{\wsync}{\ensuremath{w_{\mathrm{s}}}\xspace} 
\newcommand{\wasync}{\ensuremath{w_{\mathrm{a}}}\xspace} 
\newcommand{\period}{\ensuremath{k}\xspace}             
\newcommand{\rsync}{\ensuremath{R_{\mathrm{s}}}\xspace}   
\newcommand{\rasync}{\ensuremath{R_{\mathrm{a}}}\xspace}  
\newcommand{\syncslot}{synchronous slot\xspace}
\newcommand{\asyncslot}{asynchronous slot\xspace}
\newcommand{\syncslots}{synchronous slots\xspace}
\newcommand{\asyncslots}{asynchronous slots\xspace}

\newcommand{\controlslot}{control slot\xspace}
\newcommand{\controlslots}{control slots\xspace}

\newcommand{\controlverdict}{control verdict\xspace}
\newcommand{\controlverdicts}{control verdicts\xspace}
\definecolor{syncrule}{HTML}{0072B2}        
\definecolor{asyncrule}{HTML}{009E73}       
\definecolor{controlreading}{HTML}{B07A00}  
\newcommand{\outputcolored}[1]{#1}          
\newcommand{\controlcolored}[1]{#1}           
\newcommand{\outputreading}{\outputcolored{output reading}\xspace}
\newcommand{\controlreading}{\controlcolored{control reading}\xspace}

\newcommand{\ymark}{\ensuremath{\checkmark}\xspace}
\newcommand{\nmark}{\ensuremath{\times}\xspace}
\newcommand{\pundecided}{\ensuremath{\bot}\xspace}
\newcommand{\pdag}{\text{DAG}\xspace}
\newcommand{\pgst}{\text{GST}\xspace}
\newcommand{\algvar}[1]{\texttt{#1}\xspace}   
\newcommand{\algfunc}[1]{\textsc{#1}\xspace}  

\newcommand{\pinterval}{\algvar{interval}}

\newcommand{\ptrycommit}{\algfunc{TryCommit}}
\newcommand{\ptrydirectdecide}{\algfunc{TryDirectDecide}}
\newcommand{\ptryindirectdecide}{\algfunc{TryIndirectDecide}}
\newcommand{\pupdateperiod}{\algfunc{UpdatePeriod}}
\newcommand{\pgetsubdag}{\algfunc{GetSubDag}}
\newcommand{\pgetleader}{\algfunc{GetLeader}}
\newcommand{\pwaveof}{\algfunc{WaveLength}}

\newcommand{\pfindanchor}{\algfunc{FindAnchor}}
\newcommand{\preplay}{\algfunc{Replay}}
\newcommand{\pproberate}{\algfunc{ProbeRate}}
\newcommand{\pcanary}{\algvar{canary}}
\newcommand{\pmaxperiod}{\algvar{maxPeriod}}

\title{\sysname: Interleaving Partially Synchronous and Asynchronous Commit Rules on a Shared DAG}

\ifpublish
    \author{
        Zeno de Angeli\inst{2} \and
        Philipp Jovanovic\inst{1,2} \and
        Lefteris Kokoris-Kogias\inst{1} \and
        Markus Legner\inst{1} \and
        Alberto Sonnino\inst{1,2}
    }
    \institute{Mysten Labs \and University College London (UCL)}
    \authorrunning{Z. de Angeli et al.}
\else
    \author{}
    \institute{}
\fi

\begin{document}

\maketitle
\ifpublish\else
    \pagestyle{plain}
    \thispagestyle{plain}
\fi

\begin{abstract}
  Dual-mode consensus protocols are fast when the network is partially synchronous and remain live under asynchrony. We introduce \sysname, a dual-mode mechanism that composes a partially synchronous and an asynchronous commit rule over one DAG: every $\period$-th round is decided by the asynchronous rule, whose leader a common coin reveals after the votes, and all other rounds by the partially synchronous rule. Every interval, validators replay the committed DAG under each candidate period, adopt the one with the fewest expected message delays, and fall back to $\period = 1$ when the output stalls; the asynchronous rule applied to the coin rounds alone keeps the protocol live. \sysname sends no message beyond the DAG's blocks, not even to agree on the period, and opens a coin only on the rounds that need a hidden leader. It is generic over pairs of DAG commit rules that share a committee; we instantiate it with \mysticeti~\cite{mysticeti} and \mahimahi~\cite{mahimahi} at $n \ge 3f+1$ and with the two variants of BlueBottle~\cite{bluebottle} at $n \ge 5f+1$. We prove it safe and live, and provide mechanized proofs in Lean~4. In simulation, \sysname matches the partially synchronous protocol in a healthy network, stays close to the asynchronous one when network conditions stall the partially synchronous one, and adapts quickly in both directions.
\end{abstract}

\section{Introduction}
\label{sec:introduction}

Partially synchronous BFT consensus~\cite{shoalpp,mysticeti,tendermint,pbft,shoal,hotstuff} is fast but loses liveness under asynchrony or a targeted leader attack~\cite{consensus-dos}. Asynchronous BFT~\cite{vaba,narwhal,mahimahi,dag-rider,honeybadger} stays live but pays extra message delays and a common coin. Dual-mode protocols~\cite{abraxas,icarus,parbft,ipotane,jolteon-ditto,tockowl,bdt,bullshark} aim for both.

Existing dual-mode designs all pay for the mode decision, in one of three ways (\Cref{tab:dualmode}). Serial designs, such as Ditto~\cite{jolteon-ditto} or Bolt-Dumbo Transformer~\cite{bdt}, leave the fast path on a local timer and must then agree on where the abandoned path stopped; this reconciliation is slow and fragile, as the attack of Rambaud et al.~\cite{rambaud-2pac} on Ditto's fallback showed. Parallel designs run both paths at once and pay quadratic messages even under synchrony: ParBFT1~\cite{parbft} settles the race between them with a separate binary agreement per slot, and Ipotane~\cite{ipotane} folds that agreement into its asynchronous path at the price of one extra round per slot. Bullshark~\cite{bullshark} instead moves the mode into the blocks, chosen by timeouts, and opens a coin for every leader whether the network is synchronous or not, which it can afford only because it is already slow: on its certified \pdag a commit takes six message delays, against three on an uncertified one~\cite{mysticeti,cordial-miners}. Icarus~\cite{icarus} abandons the second path altogether, rotating instead among the validators' chains, but still runs a Byzantine agreement at every switch to align what the previous path committed. We discuss these designs in \Cref{sec:related}.

\sysname makes no such decision. DAG commit rules only read the \pdag and, beyond pacing, never shape it, so two rules can share one \pdag: each round holds one leader slot, and the round number picks which rule decides it. Every \period-th slot is an \asyncslot, decided by the asynchronous rule (\rasync), and the others are \syncslots, decided by the partially synchronous one (\rsync). There is no mode vote, no path switch, and no extra message. Because reading the DAG is virtually free, the same blocks can also be read a second time, for a different purpose (\Cref{sec:adaptive-control}). The period \period acts as a dial, with $\period = 1$ giving an asynchronous protocol and $\period = \infty$ giving a partially synchronous one (\Cref{fig:idea}).

\paragraph{Insights.}
First, each rule needs a fixed number of rounds to decide a slot, called \emph{wavelength}; the two rules have different wavelengths. As in every DAG protocol, a slot left undecided after its wavelength is decided later by the first committed slot above it, its \emph{anchor}. \sysname starts the anchor search one wavelength above the undecided slot, using the wavelength of the undecided slot rather than that of the anchor. Any anchor at that height sees evidence of a direct commit, whichever rule produced it. This is the core of the safety argument. It lets \sysname skip the reconciliation step of related work~\cite{jolteon-ditto,abraxas,parbft}, agreeing on where the previous rule stopped: the slots one rule leaves pending are decided by the anchors above them, whichever rule decides those.

Second, because a change is completed by the anchors that follow it, a change of rule is safe whenever it happens. \sysname thus never has to decide a mode switch from the local observation of a number of timeouts, which is how Bullshark~\cite{bullshark} and serial designs choose their mode~\cite{jolteon-ditto,bdt}.

Third, every \period-th slot is decided by the asynchronous rule, so the \pdag carries a coin every \period rounds rather than every round. Under asynchrony an adversary can keep the known-leader slots undecided, and output stops behind the first of them. Validators therefore read the same blocks a second time, applying the asynchronous rule to the coin-carrying slots alone. This \emph{\controlreading} outputs nothing, but it keeps deciding under asynchrony, and every honest validator obtains from it the same committed block from which to compute the period (\Cref{sec:adaptive-control}). In the opposite direction, at $\period = 1$ nobody waits for a known leader, so no evidence about the partially synchronous rule would form; on a few scheduled \emph{canary rounds}, validators still wait for it, and the certificates that form are that evidence (\Cref{sec:adaptive-probes}). The period is thus a deterministic function of data every honest validator holds identically, as in Shoal~\cite{shoal} and Hammerhead~\cite{hammerhead}, and drops to $1$ when output has stalled.

Finally, validators replay the agreed data under each candidate period and keep the period that would have committed with the lowest latency (\Cref{sec:adaptive-replay}). Rounds that ran under the partially synchronous rule opened no coin, so \sysname averages over the $n$ possible leaders, which is the expectation a uniform coin would give. This is sound because the replay only ranks the two rules and commits nothing. The asynchronous rule certifies a fixed fraction of every round's blocks under any message schedule, so an adversary cannot drive that average down.

We instantiate \sysname on two pairs of uncertified DAG protocols. As an example of $n \ge 3f + 1$, we pair \mysticeti~\cite{mysticeti} with \mahimahi~\cite{mahimahi}. At $n \ge 5f + 1$, we pair the partially synchronous and asynchronous variants of BlueBottle~\cite{bluebottle}, which that paper presents as two separate protocols.

\paragraph{Contributions.}
\begin{itemize}
    \item An adaptive period selection mechanism allowing validators to recompute the best commit rule at run time from the committed \pdag alone, with no extra message and no timer.
    \item \sysname, a dual-mode mechanism in which two commit rules interpret one \pdag and the round number alone selects which of them decides each slot.
    \item Proofs of safety and liveness, machine-checked in Lean~4.
    \item An implementation of \sysname, over two types of protocols, $n \ge 3f + 1$ (\mysticeti and \mahimahi) and $n \ge 5f + 1$ (BlueBottle) on top of the authors' codebase, and simulations in various network conditions.
\end{itemize}

\begin{figure}[t]
    \centering
\begin{tikzpicture}[
        font=\scriptsize,
        slot/.style={draw, line width=0.4pt, text=white, font=\bfseries\scriptsize, minimum size=0.38cm, inner sep=0pt},
        sync/.style={slot, rectangle, draw=syncrule, fill=syncrule},
        async/.style={slot, circle, draw=asyncrule, fill=asyncrule},
        outer/.style={slot, rectangle, draw=syncrule, fill=syncrule, minimum size=0.48cm},
        inner/.style={slot, circle, draw=asyncrule, fill=asyncrule, minimum size=0.3cm},
        label/.style={anchor=east, inner sep=0pt, text=black},
        wait/.style={anchor=north, inner sep=1pt, font=\tiny, text=syncrule},
    ]
    \node[label] at (0.350, 0.000) {$\period = 8$};
    \node[sync] at (0.640, 0.000) {1};
    \node[sync] at (1.280, 0.000) {2};
    \node[sync] at (1.920, 0.000) {3};
    \node[sync] at (2.560, 0.000) {4};
    \node[sync] at (3.200, 0.000) {5};
    \node[sync] at (3.840, 0.000) {6};
    \node[sync] at (4.480, 0.000) {7};
    \node[async] at (5.120, 0.000) {8};
    \node[sync] at (5.760, 0.000) {9};
    \node[sync] at (6.400, 0.000) {10};
    \node[sync] at (7.040, 0.000) {11};
    \node[sync] at (7.680, 0.000) {12};
    \node[sync] at (8.320, 0.000) {13};
    \node[sync] at (8.960, 0.000) {14};
    \node[sync] at (9.600, 0.000) {15};
    \node[async] at (10.240, 0.000) {16};
    \node[label] at (0.350, -0.750) {$\period = 1$};
    \node[async] at (0.640, -0.750) {1};
    \node[async] at (1.280, -0.750) {2};
    \node[async] at (1.920, -0.750) {3};
    \node[async] at (2.560, -0.750) {4};
    \node[async] at (3.200, -0.750) {5};
    \node[async] at (3.840, -0.750) {6};
    \node[outer] at (4.480, -0.750) {};
    \node[inner] at (4.480, -0.750) {7};
    \node[wait] at (4.480, -1.010) {wait};
    \node[async] at (5.120, -0.750) {8};
    \node[async] at (5.760, -0.750) {9};
    \node[async] at (6.400, -0.750) {10};
    \node[async] at (7.040, -0.750) {11};
    \node[async] at (7.680, -0.750) {12};
    \node[async] at (8.320, -0.750) {13};
    \node[outer] at (8.960, -0.750) {};
    \node[inner] at (8.960, -0.750) {14};
    \node[wait] at (8.960, -1.010) {wait};
    \node[async] at (9.600, -0.750) {15};
    \node[async] at (10.240, -0.750) {16};
\end{tikzpicture}
    \caption{Slots by round number. Squares are \syncslots, decided by \rsync; circles are \asyncslots, decided by \rasync. At $\period = 8$, seven slots are decided by \rsync and the eighth carries a coin; at $\period = 1$ every slot is an \asyncslot. A canary round, drawn as both shapes, is an \asyncslot on which validators still wait for the known leader.}
    \label{fig:idea}
\end{figure}
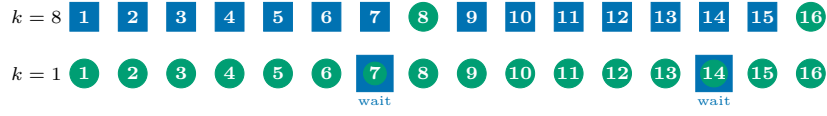

\section{Background and Model}
\label{sec:model}

\paragraph{System model.} \label{sec:model-system}
The system consists of $n$ validators, of which at most $f$ are Byzantine, and every quorum has size $n - f$; each instantiation fixes the resilience (in our evaluation either $n \ge 3f + 1$ or $n \ge 5f + 1$). Validators communicate over point-to-point channels, and the DAG substrate carries the base protocols' standard assumptions of signed blocks and hash-based causal references. \sysname itself adds no cryptography and no messages.

Under partial synchrony~\cite{dls}, message delays between honest validators are bounded by a known $\Delta$ after \pgst. Under asynchrony~\cite{flp}, an adaptive adversary schedules delivery but must eventually deliver every message.
For a coin-based \rasync we assume a common coin as a black box~\cite{dag-rider,mahimahi,bullshark}, with three properties: agreement, unpredictability until $n - f$ authorities have contributed, and uniformity conditional on everything the adversary has seen before the value can be learned, including earlier coins.
The round's leader is a balanced map of the coin's value into the $n$ authorities, so it is uniform and unknown until $n - f$ shares exist. The coin can be instantiated with threshold signatures~\cite{cks}, one share per block of the rounds that need a coin.

\paragraph{DAG-based consensus.} \label{sec:model-dag}
In an uncertified DAG there is at most one block per honest validator per round, and each block references $n - f$ blocks of the previous round. A validator advances a round after receiving $n - f$ blocks from the previous round, and after a leader wait (in partial synchrony). Equivocation is tolerated rather than prevented: a vote is the first-seen block per author and round in a depth-first search.
Each round has a known number of leader slots. The direct rule decides each slot as a \emph{commit} or a \emph{skip} from the rounds of its wave, or leaves it \pundecided. Under the indirect rule, an \pundecided slot is decided by its \emph{anchor}, the earliest later slot that is committed or \pundecided, passing over skipped slots. The slot commits if and only if the anchor's causal history holds a certificate for it~\cite{mysticeti}. Output proceeds in slot order, so an \pundecided slot holds back every slot above it; each committed leader adds to the output, in a deterministic order, the blocks of its causal history not already in that of an earlier committed leader.

A rule of wavelength \wl decides the slot at round $r$ from the wave of rounds $r, \dots, r + \wl - 1$ alone: it commits when the blocks of that wave reference the leader block in the pattern the rule prescribes, skips when enough of them fail to, and otherwise leaves the slot \pundecided. \sysname never inspects that pattern: it calls the rule's own decision predicates and reads only the parameters of \Cref{tab:rules}, so any rule of this shape plugs in unchanged. As an example, in \mysticeti and \mahimahi round $r + \wl - 2$ is the vote round and round $r + \wl - 1$ the certify round, a certificate is a certify-round block referencing $n - f$ votes, a direct commit needs $n - f$ certificates and a direct skip $n - f$ non-votes, and the boost rounds fill the gap at $\wl \ge 4$.

\begin{table}[t]
    \centering
    \small
    \caption{The four base protocols and the parameters \sysname reads}
    \label{tab:rules}
    \scriptsize
    \begin{tabular*}{\textwidth}{@{\extracolsep{\fill}}lclcc@{}}
        \toprule
        Rule                & Wavelength (\wl) & Leader election & Leader wait  & Committee      \\
        \midrule
        \mysticeti          & $3$              & Deterministic   & \ymark       & $n \ge 3f + 1$ \\
        \mahimahi           & $\{4, 5\}$       & Randomized      & \nmark       & $n \ge 3f + 1$ \\
        \midrule
        BlueBottle (sync.)  & $2$              & Deterministic   & \ymark       & $n \ge 5f + 1$ \\
        BlueBottle (async.) & $3$              & Randomized      & \nmark       & $n \ge 5f + 1$ \\
        \bottomrule
    \end{tabular*}
\end{table}

\paragraph{What the two rules must provide.} \label{sec:interface}
\sysname requires five (typical) properties from the two rules it is instantiated with:

\begin{itemize}
    \clauseitem{clause:a1} Both rules run on the same committee, with quorums of size $n - f$, and read the same blocks; an honest author produces one block per round, referencing a quorum of the previous round and every block it holds that its previous block did not cover, and blocks carry whatever either rule needs.
    \clauseitem{clause:a2} A rule decides the slot at round $r$ from the wave of rounds $r, \dots, r + \wl - 1$, and a direct commit leaves its evidence in the causal history of every block at round $r + \wl$ or above.
    \clauseitem{clause:a3} A direct skip at a slot implies that no certificate for it ever forms.
    \clauseitem{clause:a4} After \pgst, with pacing, a \syncslot whose leader is honest and whose successor round carries a leader wait is directly committed.
    \clauseitem{clause:a5} Before the round's coin can be learned, the revealed history fixes a set of at least $pn$ candidate authors whose blocks have direct-commit evidence in every admissible continuation that populates the wave, the counting property; conditional on that history the leader is uniform over all authors, so a direct commit has probability at least $p$ whatever the adversary learned from earlier coins; and each candidate's evidence is checkable on a committed window.
\end{itemize}

The wave family of \Cref{tab:rules} provides \crefrange{clause:a1}{clause:a3} by construction: evidence accumulates in a wave, quorum intersection carries a direct commit's evidence into every block one wave later, and the commit and skip thresholds intersect, so a skipped slot can never be certified. A pair must add one clause each, \cref{clause:a4} for the partially synchronous rule and \cref{clause:a5} for the asynchronous one, which is why no single protocol satisfies the interface alone. Each clause is discharged for the two pairs we instantiate (\Cref{app:full-proofs}).

\paragraph{Problem statement.} \label{sec:model-problem}
\sysname composes a rule \rsync, of wavelength \wsync and with a known leader, and a rule \rasync, of wavelength \wasync and with a coin-elected leader, into one protocol that inherits the latency of \rsync under periods of synchrony and the liveness of \rasync under asynchrony. Both rules are used through the interface above. \sysname implements atomic broadcast.

\begin{definition}[Atomic broadcast]\label{def:bab}
    Atomic broadcast satisfies:
    \begin{itemize}
        \item \textbf{Agreement:} If one honest validator commits a block, all honest validators eventually commit it.
        \item \textbf{Integrity:} Each honest validator commits a block at most once, and only if it was proposed by its author.
        \item \textbf{Validity:} If an honest validator proposes a block, all honest validators eventually commit it (with probability $1$).
        \item \textbf{Total order:} If an honest validator commits a block $b$ before a block $b'$, no honest validator commits $b'$ before $b$.
    \end{itemize}
\end{definition}

\section{Interleaving Two Commit Rules}
\label{sec:protocol}\label{sec:overview}

\sysname adjusts the period \period at run time. For clarity, we first describe the protocol at a fixed \period in this section, then describe the full adaptive protocol in \Cref{sec:adaptive}. The reader may keep \mysticeti and \mahimahi in mind throughout, to make the exposition concrete.

\subsection{Slots by Round Number}
\label{sec:interleave-schedule}

The wavelength function is $\wlof{r} = \wasync$ if $r \bmod \period = 0$, and \wsync otherwise, where \period is the period in force at round $r$ (\Cref{sec:adaptive} indexes it by interval, as $\period_{j(r)}$, once it adapts). Each slot is therefore decided by exactly one rule, fixed by its round number, and this function is the only thing \sysname adds to the decision rule. Blocks are identical in every round, aside from a coin share, and the \pdag never pauses; the wavelength dictates only which rounds hold a slot's evidence. \Cref{fig:running} is the running example of \Cref{sec:protocol,sec:adaptive}.

\begin{figure}[tb]
    \centering
    \input{figures/running}
    \caption{The two rules deciding one another's slots, on one \pdag at $\period = 4$ ($n = 4$, $\wsync = 3$, $\wasync = 5$). Squares are \syncslots and circles \asyncslots, and the brackets mark the wave each rule reads, three rounds for \rsync and five for \rasync; filled is commit, white is undecided. Slots $8$, $10$ and $11$ are left undecided because their waves reach above the drawn rounds. The direct rule leaves slot $4$ undecided: only one block of round $8$ certifies its coin leader, short of the three a direct commit needs. Slot $9$, the first slot at its anchor floor $4 + \wasync$, commits directly under \rsync, and its causal history holds that certificate (the arrow), so it decides slot $4$.}
    \label{fig:running}
\end{figure}

A \syncslot's leader is known in advance from a deterministic schedule, such as a round-robin or reputation-based schedule~\cite{shoal,hammerhead}. An \asyncslot's leader is revealed by the coin only once its votes are fixed, at round $r + \wasync - 1$.
Validators wait, up to a bounded timeout, for a \syncslot's leader block before proposing at the round above it and, in \mysticeti, for a quorum of votes for it before proposing one round higher. Nobody waits for the hidden leader of an \asyncslot; the canary rounds of \Cref{sec:adaptive-probes} are the one exception. \Cref{clause:a4} rests on these waits, which \sysname inherits unchanged; a \syncslot immediately below an \asyncslot lacks the second one, so it may be decided by its anchor instead of directly, and \cref{clause:a4} is stated for \syncslots whose successor round carries a leader wait, which the canary rounds restore even at $\period = 1$.

\subsection{Deciding Slots: the Anchor Floor}
\label{sec:interleave-safety}

Each slot is decided by its own rule over its own wave (\Cref{alg:decide}), and this interpretation of the \pdag, which considers every slot and produces the ledger, is the \outputreading; \Cref{sec:adaptive-control} adds a second reading of the same blocks later. The anchor of the slot at round $r$ is searched from round $r + \wlof{r}$, a floor set by the wavelength of the \emph{undecided slot}, not by that of the anchor. A direct commit leaves certificates at round $r + \wlof{r} - 1$, so every block from round $r + \wlof{r}$ on holds one in its causal history. An anchor closer than that, such as an \asyncslot anchored at the \syncslot right above it, cannot see those certificates and would skip a slot that another validator committed. \Cref{fig:running} shows this: the direct rule leaves the \asyncslot at round $4$ undecided, and its anchor search starts at round $4 + \wasync$, where the \syncslot at round $9$ has committed. \Cref{app:full-proofs} discusses and illustrates (\Cref{fig:anchor}) the choice of floor in further detail.

The consequence is handover (\Cref{cor:handover}, \Cref{app:full-proofs}): the indirect step of either rule completes a direct commit of the other that it never witnessed. The arrow in \Cref{fig:running} marks this: the \syncslot's commit decides a slot the asynchronous rule left undecided. This is the whole cross-rule argument, and the reason a period change never has to pause the protocol to finish pending slots.

\subsection{What a Fixed Period Cannot Do}
\label{sec:interleave-limits}

In a healthy network, an \asyncslot costs $\wasync - \wsync$ extra rounds once every \period rounds, and its successor waits at most $\max(0, \wasync - \wsync - 1)$ rounds for causal ordering. These delays never compound and are negligible when the network is synchronous and \period is large. \Cref{app:variant} discusses a variant that avoids even this, at the price of a longer proof.

Under asynchrony, every \syncslot could be left \pundecided. The \asyncslots still commit, because their leader is hidden, but nothing they commit is output: slots are output in round order, and the anchor search of an \pundecided \syncslot stops at the next \pundecided one. No fixed $\period > 1$ is therefore live under asynchrony, while $\period = 1$ is live but slow when the network is healthy. Consequently, the period must be adjusted dynamically, which is the subject of \Cref{sec:adaptive}.

\begin{algorithm}[t]
    \caption{Decision with a slot-dependent wavelength}
    \label{alg:decide}
    \begin{algorithmic}[1]
        \Procedure{\pwaveof}{$r$}
        \State \Return $\wasync$ \textbf{if} $r \bmod \period = 0$ \textbf{else} $\wsync$ \Comment{\period in force for round $r$}
        \EndProcedure
        \Statex
        \Procedure{\ptrycommit}{$\pdag$}
        \State \textbf{for} each slot $s$ at round $r$, from the highest undecided downward: \Comment{as \mysticeti}
        \State \quad $\wl \gets$ \Call{\pwaveof}{$r$}
        \State \quad $v \gets$ \Call{\ptrydirectdecide}{$s$, $\wl$} \Comment{skip on $n{-}f$ blames or commit on $n{-}f$ certificates}
        \State \quad \textbf{if} $v = \pundecided$: $v \gets$ \Call{\ptryindirectdecide}{$s$, $\wl$}
        \State \Return the maximal decided prefix, in round order
        \EndProcedure
        \Statex
        \Procedure{\ptryindirectdecide}{$s$, $\wl$}
        \State $a \gets$ \Call{\pfindanchor}{$s$, $\wl$} \Comment{earliest commit-or-undecided slot at round $\ge r + \wl$}
        \State \textbf{if} $a$ is undecided: \Return $\pundecided$
        \State \textbf{if} a certificate for $s$'s leader block is in $a$'s causal history: \Return commit
        \State \textbf{else}: \Return skip
        \EndProcedure
    \end{algorithmic}
\end{algorithm}

\section{Adapting the Period}
\label{sec:adaptive}

The period selects between the two rules from what the committed \pdag shows, rather than from a guess about the network, which the adversary controls. \Cref{sec:adaptive-loop} states the loop, and the three subsections after it supply what the loop needs: evidence of both rules in the \pdag (\Cref{sec:adaptive-probes}), agreement on that evidence even when the ledger is stuck (\Cref{sec:adaptive-control}), and a score that turns it into a period (\Cref{sec:adaptive-replay}).

\begin{algorithm}[t]
    \caption{The adaptation loop, abstract form (runs at every validator)}
    \label{alg:loop}
    \begin{algorithmic}[1]
        \State $\pinterval$, $\pmaxperiod$, $\pcanary$, $\epsilon$ \Comment{e.g., $128$ rounds, period $64$, canary $31$, hysteresis $2\%$}
        \State $\period \gets \pmaxperiod$ \Comment{state: the period in force}
        \Statex
        \Procedure{OnInterval}{$j$} \Comment{once the scan of interval $j$ ends}
        \State $A \gets$ the pivot: the first commit of the \controlreading in interval $j$ \Comment{\Cref{sec:adaptive-control}; none keeps \period}
        \State $W \gets$ causal history of $A$ over the preceding \pinterval rounds \Comment{the window: agreed data}
        \If{the ledger up to $A$ committed no slot in $W$} \Comment{failover on a stalled ledger}
        \State $\period \gets 1$
        \Else
        \State $\period^\star \gets \arg\min_{\period' \in K} \Call{\preplay}{$W$, $\period'$}$ \Comment{expected output delay under each candidate (\Cref{sec:adaptive-replay})}
        \State \textbf{if} $\period^\star$ beats \period by more than $\epsilon$: $\period \gets \period^\star$ \Comment{hysteresis}
        \EndIf
        \State \Return $\period$ \Comment{applies to the rounds of interval $j + 1$ and above}
        \EndProcedure
    \end{algorithmic}
\end{algorithm}

\subsection{The Adaptation Loop}
\label{sec:adaptive-loop}

Rounds are grouped into intervals of \pinterval rounds, each running under one period. In each interval a \emph{pivot} is chosen, the first commit of the \controlreading, a second reading of the same blocks that applies the asynchronous rule to the rounds that carry a coin and whose verdicts never enter the ledger (\Cref{sec:adaptive-control}). The causal history of that pivot over the last \pinterval rounds is the \emph{window}, and the period of the next interval is the candidate whose replay of the window scores best (\Cref{sec:adaptive-replay}). \Cref{alg:loop} states the loop. As a failover, if the ledger up to the pivot committed nothing in the window, the next period is $1$ outright: the replay ranks periods, while the failover guarantees that an output which has stopped advancing always gets the rule that makes it advance again (and greatly simplifies the liveness proof).

The new period is a deterministic function of the pivot's window and of round numbers, and it applies only from the next interval on. By induction over intervals, every honest validator computes the same sequence of periods, for any deterministic update rule (\Cref{thm:agreement-period}). No message and no vote is involved.

The remaining parameters and guards are of no consequence to the argument and are deferred to \Cref{app:full-proofs}: the gating rule, under which a slot is evaluated once its interval's period is known, the one-interval lag, the bounds on \pinterval and \pmaxperiod, the warm-up interval, the hysteresis, and the candidate set $\{1, 2, 4, \dots, \pmaxperiod\}$.

\subsection{Seeing Both Rules: Two-Way Probes}
\label{sec:adaptive-probes}

A window shows only what the running rule did. Without a probe, at $\period = 1$ nothing would wait for a known leader, so the partially synchronous rule would look dead forever and the period could never climb back; and at $\period = \infty$ there would be no coin at all, so nothing would be live when the network turns, which is why $\infty$ is not a candidate period.
Each regime---the synchronous one, at a large \period, and the asynchronous one, at $\period = 1$---therefore keeps a thin probe of the other rule alive, at a fixed cost. There are two dials, one per direction.

\para{From \rsync to \rasync} In the synchronous regime, the \asyncslot every \period rounds (with $\period \le \pmaxperiod$) is the probe of the asynchronous rule. It carries a coin, it commits under asynchrony, and it keeps the \controlreading moving, so the period can fall when the network turns.

\para{From \rasync to \rsync} In the asynchronous regime, every round $r$ with $r \bmod \pcanary = 0$ is a \emph{canary}: validators keep the leader wait for the known leader of that round even on an \asyncslot. A canary that gathers a certificate is evidence that the partially synchronous rule would work again, so the period can climb back. The cost is one leader wait every \pcanary rounds (although wall-time waits are meaningless in asynchrony), with \pcanary odd and hence coprime to every candidate period, so that every candidate is probed.\footnote{A \pcanary sharing a factor with some candidate would leave all its canaries on that candidate's \asyncslots, and the period could never climb to it.} The same wait restores the precondition of \cref{clause:a4} at $\period = 1$ (\Cref{sec:interleave-schedule}), so the canaries serve twice.

Neither regime ever loses sight of the other rule, and switching regimes is reading that evidence. The two probe rows of \Cref{fig:idea} show it.

\subsection{Agreeing on the Evidence: the Control Reading}
\label{sec:adaptive-control}

The pivot cannot come from the ledger. Under asynchrony the ledger stalls below the first \syncslot left \pundecided, so a pivot taken from it would never be committed and the period could never fall (\Cref{fig:readings}).
The \controlreading considers only its own slots, defined next, and searches anchors among them alone. Each carries a coin, so its leader is hidden and asynchrony cannot stall the reading, and its first commit in an interval is the pivot.

\begin{figure}[tp]
    \centering
    \input{figures/readings}
    \caption{One \pdag read twice ($n = 4$, $\wsync = 3$, $\wasync = 5$, $\period = 4$, an interval of $8$ rounds ending at the dotted line, every known leader's block reaching only $f + 1$ validators in time). Squares are \syncslots and circles \asyncslots, filled when the reading commits them, struck on a skip and outlined while \pundecided; the brackets mark a wave. \textbf{Top:} the \outputreading; slots $4$ and $8$ commit, but slots $2$ and $3$ wait at slot $6$ (arrows), so nothing is output. \textbf{Bottom:} the \controlreading of the same blocks considers the coin-carrying rounds alone and resolves anyway; its first commit, slot $4$, is the pivot, and an output that committed nothing sets $\period = 1$ above round $8$. Slots $4$ and $8$ belong to both readings and are committed directly in both.}
    \label{fig:readings}
\end{figure}

Its slots are the interval's \asyncslots, the rounds with $r \bmod \period = 0$, and, above the interval, the multiples of \pmaxperiod. The anchor of a slot near the end of an interval lies above it, where the period is precisely what this scan is computing, so $\period$ is no guide there; the multiples of \pmaxperiod are \asyncslots under every candidate period, hence defined and coin-carrying whatever the scan decides. Because this set is fixed by rule rather than by what a view holds, the reading is agreed by the same lemmas as the ledger, with direct verdicts identical in both readings.

\sysname opens one coin every \period rounds, and the \controlreading opens none of its own: every candidate divides \pmaxperiod, so the multiples of \pmaxperiod that the reading uses above an interval boundary already carry one. They are what producers fall back on while an interval's period is still being computed and they cannot yet tell which of its rounds are \asyncslots; a share for such a round is contributed late.

The two readings overlap and need not agree: a round divisible by the period is a slot of both. Where either reading decides directly, both decide alike (\Cref{thm:agreement,cor:handover}, \Cref{app:full-proofs}); only their indirect verdicts can differ, because each searches anchors among its own slots. That divergence is harmless, since the two answer different questions: the \outputreading says what is committed, and the \controlreading says only where the period is computed from. A reading is a way to analyze an agreed subset of a \pdag whose full interpretation is stuck.

\subsection{Scoring Without a Coin: Counterfactual Replay}
\label{sec:adaptive-replay}

Because a commit rule only reads the \pdag, the window can be reinterpreted under any candidate period, with $\wl'(r) = \wasync$ if $r \bmod \period' = 0$ and \wsync otherwise. The score is the expected delay from a round to the output of its blocks, so the head-of-line blocking caused by \pundecided slots is already in it and nothing is added by hand. Each rule is scored from the evidence the window holds about it, which is a different kind of evidence for each.

\para{The asynchronous rule, without a coin} Its only random input is a uniform leader, so its expected latency on the window is the plain average over the $n$ possible leaders, each decided from the certificates the window holds. The slot at round $r$ commits directly with probability $c_r / n$, where $c_r$ is the number of the $n$ candidate leaders that the window shows directly committed, and no coin is reconstructed.

\para{The partially synchronous rule, from the canaries} On an \asyncslot only the canary rounds waited for the known leader, so the replay takes its evidence from the canary rounds alone. Their success rate stands in for the candidate's other \syncslots, and the replay is exact when the window holds no probe.

The adversary cannot make the asynchronous rule look slow: by the counting property, $c_r / n$ has a floor under any scheduling (\Cref{lem:replay-starvation}; at least $1/3$ for the $3f + 1$ pair at $\wasync = 5$). It can serve the public canaries and starve the other known leaders: if the output then stops, the failover of \Cref{sec:adaptive-loop} takes over, and otherwise the output keeps advancing. Behaving well in one window and badly in the next is a regime change, measured by the next window.

The replay makes two approximations, both deterministic and identical for every candidate. First, where the rule would decide a slot from the causal history of the anchor that the candidate would have used, the replay asks only whether a certificate for the slot is anywhere in the window; the window is the pivot's own causal history, so it stands in for the anchor's. Second, a slot that a candidate never commits within the window counts as committing at the window's top round, so the rounds it holds back are deferred there rather than left undefined. The cost is $O(|K| \cdot |W|)$ for $|K|$ candidates over a window of $|W|$ blocks, and \Cref{alg:update} in \Cref{app:full-proofs} gives the procedure in full.

\section{Correctness}
\label{sec:proofs}

\sysname implements atomic broadcast (\Cref{def:bab}) for any pair of rules satisfying the interface of \Cref{sec:interface}. Full statements and proofs are in \Cref{app:full-proofs}, machine-checked in Lean~4 (\Cref{app:lean}).\footnote{\leanlink}

\paragraph{Safety.} A direct commit leaves a certificate in every block from round $r + \wlof{r}$ on, and the anchor search starts there, so an indirect decision by either rule agrees with a direct commit by the other. Two indirect decisions agree because validators find the same anchor. This is the one cross-rule argument; everything else is inherited from the base protocols (\Cref{thm:agreement}, \Cref{cor:handover,cor:order-integrity}).

\paragraph{Liveness.} After \pgst, \syncslots with an honest leader and a leader wait on their successor round commit directly by \cref{clause:a4}, and \asyncslots do not get in the way (\Cref{thm:liveness-ps}). Under asynchrony, the \controlreading keeps resolving, because every slot on it has a hidden leader. The failover then forces $\period = 1$, and at $\period = 1$ a run of \wasync direct commits, which the counting property makes occur with probability $1$, decides every slot below it, so the output resumes (\Cref{thm:liveness-async}).

\paragraph{Agreement on the period.} The period is a deterministic function of the pivot's window, and the pivot is agreed because the \controlreading is. Induction over intervals gives the same sequence of periods at every honest validator, for any deterministic update rule (\Cref{thm:agreement-period}). At $\period = 1$ and $\period = \infty$, \sysname's verdicts are exactly \mahimahi's and \mysticeti's (\Cref{thm:conservativity}).

\section{Implementation and Evaluation}
\label{sec:evaluation}\label{sec:implementation}

Our evaluation quantifies what deciding slots by round number buys over each pure protocol on the same DAG, under network conditions that favor one or the other. Benchmarking BFT protocols under actual Byzantine behavior is an open problem~\cite{twins}; worst-case guarantees are established by the formal proofs of \Cref{sec:proofs}, while this evaluation measures latency under specific network conditions. We make the following claims:
\begin{itemize}
    \claimitem{claim:good} In a healthy network, \sysname matches the latency of the partially synchronous protocol.
    \claimitem{claim:async} Under network conditions that stall the partially synchronous protocol, \sysname matches the latency of the asynchronous one.
    \claimitem{claim:faults} Under benign asynchrony, \sysname matches the latency of the best of the two protocols.
    \claimitem{claim:adaptive} \sysname adapts promptly to the network: it switches to the best protocol when conditions change and back when they lift.
    \claimitem{claim:generic} \Cref{claim:good,claim:async,claim:faults,claim:adaptive} hold for both classes of target protocols, at $n \ge 3f+1$ and at $n \ge 5f+1$.
\end{itemize}

\subsection{Implementation}
\label{sec:eval-implementation}

We instantiate \sysname on two pairs of protocols: \mysticeti~\cite{mysticeti} with \mahimahi~\cite{mahimahi}, at $n \ge 3f+1$, and BlueBottle's partially synchronous variant with its asynchronous variant~\cite{bluebottle}, at $n \ge 5f+1$. We implement both in Rust on top of the codebase of \mysticeti, \mahimahi, and BlueBottle~\cite{mysticeti-code} written by their respective authors; public leaders use a round-robin schedule.\footnote{\codelink} \sysname adds a per-round wavelength schedule read by the decision rule and the anchor search, restricts the leader wait to \syncslots and canary rounds, and implements the period update by counterfactual replay of \Cref{alg:update}. The replay module is about $770$ lines; integrating it into the committer and the protocol table takes about $700$ more. The block format remains untouched, and the DAG layer changes only in its round pacing, which reads the period to select the leader wait, and in its retention horizon: \sysname adds no protocol messages, no cryptography beyond the coin the asynchronous protocol already carries, and no storage access. It runs purely in memory, and its memory is bounded by the interval: the replay reads the blocks of the last \pinterval rounds.

\subsection{Experimental Setup}
\label{sec:eval-setup}

The network conditions under test must be identical across protocols, switch on and off at scripted times, and remain reproducible. This is not possible in a cloud deployment, where asynchrony cannot be injected on demand and never repeats identically. We therefore use a deterministic discrete-event simulator with seeded runs, virtual time, and the real consensus code. We use a committee of $n = 50$ validators ($f = 16$ for the $3f+1$ pair, $f = 9$ for the $5f+1$ pair) in a full mesh; \Cref{app:evaluation-details} repeats the runs at $n = 10$. Per-link latency is uniform in $25$--$50$\,ms. The leader timeout is $100$\,ms.

Each timeline runs a healthy network for $30$\,s, applies a specific network condition until $150$\,s, and restores the healthy network for a final $60$\,s. Plateau figures are taken over the windows between $60$ and $150$\,s. We test six conditions: (i) a small leader delay of $30$\,ms, below the timeout; (ii) a large leader delay of $125$\,ms, above it; (iii) a permanent crash of $f$ validators; (iv) a partial random delay, where each message is delayed with probability $0.3$ by a duration drawn uniformly at random in $100$--$150$\,ms; (v) a full random delay, the same with probability $1$; and (vi) a high jitter, an exponentially distributed delay of mean $75$\,ms capped at $400$\,ms. The leader delays target the public round-robin schedule and are blind to the coin; the random delays instantiate the random asynchronous model of Danezis et al.~\cite{random-network}.

We compare the two pure protocols of each pair with adaptive \sysname, and report end-to-end latency from submission to commit, as the mean per $5$\,s window and the median over seven seeded runs. \Cref{app:evaluation-details} lists every parameter of \sysname and of the simulator, and the per-panel numbers.

\begin{figure}[t]
    \centering
    \includegraphics[width=\linewidth]{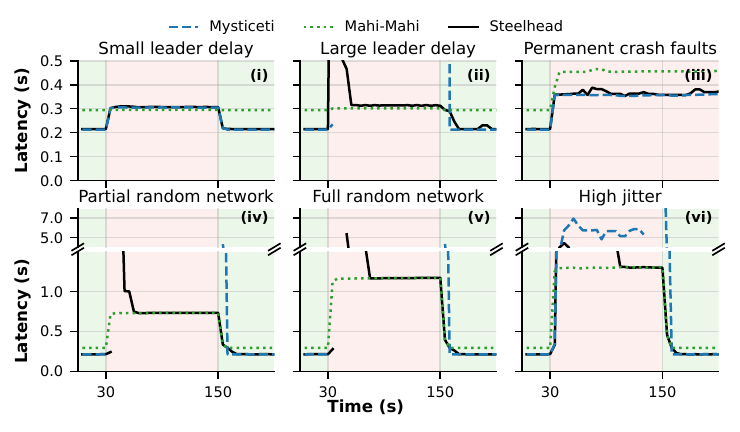}
    \caption{Latency over time of \mysticeti, \mahimahi, and adaptive \sysname ($3f+1$ pair, $n = 50$) under six network conditions applied from $30$\,s to $150$\,s (shaded). The bottom row breaks its axis: \mysticeti runs off the chart under large leader delay and full random delay.}
    \label{fig:weather-mm-50}
\end{figure}

\subsection{Healthy Network}
\label{sec:eval-healthy}

In the first $30$\,s and the last $60$\,s of every panel in \Cref{fig:weather-mm-50}, most legibly in (i), \sysname commits at $215$\,ms, \mysticeti at $213$\,ms, and \mahimahi at $294$\,ms. \sysname is within $1\%$ of \mysticeti and $27\%$ below \mahimahi. The period sits at its maximum, so one slot in $64$ operates as an \asyncslot and pays the two extra rounds described in \Cref{sec:protocol}, while everything else runs the \mysticeti rule. The counterfactual replay never lowers the period in a healthy window. This confirms \Cref{claim:good}.

\subsection{Degraded Network and Recovery}
\label{sec:eval-degraded}

\mysticeti stalls in panels (ii), (iv), and (v) of \Cref{fig:weather-mm-50}. In (ii), every known leader's block arrives after the timeout, so every slot is directly skipped and nothing commits. In (iv) and (v) it commits sporadically, by the mechanism of \Cref{app:commit-prob}, and under the jitter of (vi) it degrades to a $5.7$\,s plateau. \mahimahi is barely affected, at $301$, $731$, $1167$, and $1300$\,ms.

\sysname descends to period $1$ within $10$--$20$\,s of the onset in (ii), $15$\,s in (iv), $20$\,s in (v), and $20$--$45$\,s in (vi). Until the period drops, \sysname behaves as \mysticeti, causing the latency spike at the start of the phase. It then sits on the \mahimahi line: $313$ vs $301$\,ms in (ii), $732$ vs $731$\,ms in (iv), and $1168$ vs $1167$\,ms in (v). Panel (vi) takes longer to settle, at $1473$ vs $1300$\,ms over the plateau window: under jitter the period passes through $2$ for some $40$\,s before it reaches $1$, and the last $25$\,s of the condition run at $1301$\,ms. This confirms \Cref{claim:async}.

In (i) \mysticeti survives: all three protocols are within $5\%$ ($305$, $295$, $308$\,ms). \sysname stays at period $64$ but for at most one one-interval drop. In (iii), \sysname records $363$\,ms vs \mysticeti $356$\,ms vs \mahimahi $455$\,ms, a $2\%$ cost. Probes landing on a crashed round-robin leader read as starved slots, so the period briefly dips and climbs back every few intervals. \sysname can naturally be composed with reputation-based leader-election~\cite{shoal,hammerhead} to remove crashed leaders from the schedule and with it this effect (\Cref{sec:adaptive-probes}). This confirms \Cref{claim:faults}.

\sysname switches in both directions within a few intervals. In every panel it is back at period $64$ within $10$\,s of the condition lifting, one interval at the healthy round rate; latency follows the period each time (\Cref{app:evaluation-details}). This confirms \Cref{claim:adaptive}.

\subsection{The \texorpdfstring{$5f+1$}{5f+1} Pair}
\label{sec:eval-generic}

We evaluate the $5f+1$ pair using the same code and parameters, with BlueBottle's two variants as the pure protocols (\Cref{fig:weather-bb-50}). In a healthy network, \sysname commits at $169$\,ms, BlueBottle-PS at $168$\,ms, and BlueBottle-Async at $212$\,ms.

BlueBottle-PS stalls only in panel (ii) of \Cref{fig:weather-bb-50}, where \sysname drops to period $1$ within $10$\,s and matches BlueBottle-Async ($229$ vs $220$\,ms). Under the network-wide conditions (iv) to (vi) it survives but no longer ties its asynchronous variant at this committee size, running $48\%$, $41\%$ and $9\%$ slower; \sysname follows the faster variant within $1\%$ ($705$ vs $702$, $837$ vs $835$, and $944$ vs $935$\,ms). The link probability of \Cref{app:commit-prob} explains the gap: at $n = 50$ a validator with minimum-quorum references still references the leader with probability $0.82$, but every commit now needs $41$ votes, and the one-layer rule pays for it under random delays. In (iii) \sysname follows BlueBottle-PS within $2\%$ ($221$ vs $217$\,ms), and in (i) it follows it exactly ($238$\,ms), which leaves it $11\%$ above the asynchronous variant, the faster of the two in that panel. At $n = 10$ (\Cref{app:evaluation-details}) the picture holds with one change: BlueBottle-PS ties its asynchronous variant under the network-wide conditions, and \sysname sits between them. This confirms \Cref{claim:generic}.

\begin{figure}[t]
    \centering
    \includegraphics[width=\linewidth]{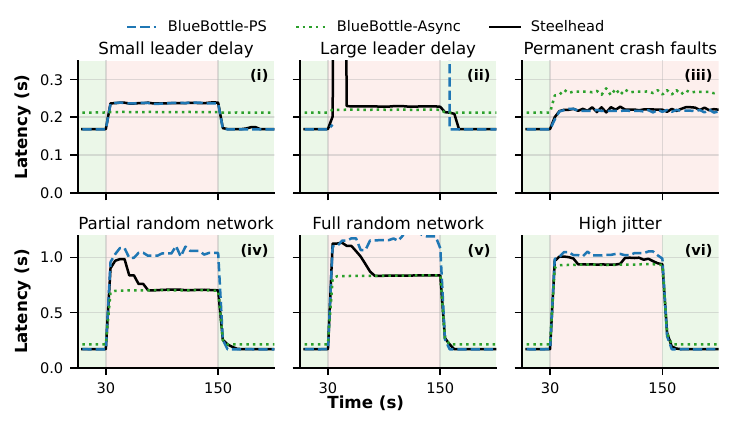}
    \caption{Latency over time of BlueBottle's partially synchronous variant (BlueBottle-PS), its asynchronous variant (BlueBottle-Async), and adaptive \sysname ($5f+1$ pair, $n = 50$) under six network conditions applied from $30$\,s to $150$\,s (shaded).}
    \label{fig:weather-bb-50}
\end{figure}
\section{Related Work}
\label{sec:related}

\begin{table}[t]
    \centering
    \caption{Dual-mode designs: what triggers a switch, and the dominant price of having one.}
    \label{tab:dualmode}
    \scriptsize
    \begin{tabular*}{\textwidth}{@{\extracolsep{\fill}}lll@{}}
        \toprule
        Design                     & Mode-change trigger                   & Dominant cost                 \\
        \midrule
        Ditto~\cite{jolteon-ditto} & quorum of timeouts        & timeouts; MVBA per switch            \\
        BDT~\cite{bdt}             & quorum of timeouts        & timeouts; binary agreement per switch \\
        Abraxas~\cite{abraxas}     & fast path lags slow chain & $O(n^2B)$ bits per block, always \\
        ParBFT1~\cite{parbft}      & first path to finish      & a VABA every slot             \\
        Ipotane~\cite{ipotane}     & first path to finish      & an extra round per slot       \\
        TockOwl+~\cite{tockowl}, ParBFT2~\cite{parbft} & hedging delay & slow-track messages when hedged \\
        Bullshark~\cite{bullshark} & quorum of timeouts       & timeouts; a coin every wave                     \\
        Icarus~\cite{icarus}       & backlog threshold         & alignment ($40$--$60\%$ when switching) \\
        \midrule
        \sysname (this work)       & round number              & a coin every \period rounds   \\
        \bottomrule
    \end{tabular*}
\end{table}

\Cref{tab:dualmode} compares these designs. Ditto~\cite{jolteon-ditto}, the Bolt-Dumbo Transformer~\cite{bdt} and Bullshark~\cite{bullshark} leave the partially synchronous path once enough parties have reported several timeouts, typically a quorum of parties timing out on multiple successive leaders. Sizing that count is the known difficulty of these designs~\cite{parbft}: too few flips a benign but poorly connected network, and too many stalls the partially synchronous path before the fallback starts. The flip is their dominant source of latency when the network degrades. The way back is not symmetric: these designs return to the fast path once the fallback commits regularly, or after a fixed number of rounds. \sysname has no timeout count to wait or size, because the round number alone says which rule decides a slot, and \period adapts from the committed \pdag. The flip itself is also paid for: Ditto runs an MVBA and the Bolt-Dumbo Transformer a binary agreement, both on the critical path, at every switch. Bullshark instead opens a coin in every wave, so every party computes and verifies a coin share for every leader whether or not the protocol runs in asynchrony. \sysname only needs a coin once every \period rounds, with $\period = 64$ or $128$ in typical deployments.

A second family avoids the timeout by never abandoning either path, and pays for the second path instead. Abraxas~\cite{abraxas} runs a slow protocol at all times alongside a fast one and uses it as the detector: parties fall back when a slow-chain block goes $\lambda$ blocks without confirmation certificates. The $O(n^2B)$ bits per block are paid in a healthy network too. ParBFT1~\cite{parbft} runs an optimistic and a pessimistic path concurrently and takes the first to finish, resolved by an asynchronous binary agreement per slot; its pessimistic path, a validated asynchronous agreement on every slot, costs a quadratic number of messages even under synchrony. Ipotane~\cite{ipotane} does the same at the cost of one extra round per slot. TockOwl+~\cite{tockowl} and ParBFT2~\cite{parbft} start their slow track after a hedging delay rather than immediately, and pay the slow track's messages only on the slots they hedge. Icarus~\cite{icarus} removes the second path instead: the chains of all parties take turns as the single optimistic path, and a switch is triggered once another chain accumulates a backlog of uncommitted blocks. Parties observe that backlog locally and may switch at different heights, so a Byzantine agreement runs on the critical path at every switch to align them, which accounts for 40--60\% of a switch. \sysname runs no second path and wastes no block: one slot in \period pays $\wasync - \wsync$ extra rounds, and nothing else.


\ifpublish
    \section*{Acknowledgements}
This work is partially funded by Mysten Labs.
We thank George Danezis for insightful discussions and for reviewing the Lean~4 formalization of \sysname (\Cref{app:lean}).

\fi

\bibliographystyle{splncs04}
\bibliography{references}

\appendix
\crefalias{section}{appendix}
\crefalias{subsection}{appendix}

\section{Variant Without Asynchronous Slots in the Ledger}
\label{app:variant}

While $\period > 1$, the \outputreading could apply \rsync to every round so that no \asyncslot ever enters the ledger. The rounds that carry a coin would be read as \asyncslots by the \controlreading alone. The failover of \Cref{sec:adaptive-loop} would still drop the period to $1$, and at $\period = 1$ the \outputreading would be \rasync on every round as before.

The liveness argument of \Cref{thm:liveness-async} uses only the \controlreading and the failover; the \asyncslots of the \outputreading play no part in it. Under asynchrony nothing they commit is output anyway (\Cref{sec:interleave-limits}), since the anchor search of an \pundecided \syncslot below stops at the next \pundecided one. The variant would therefore save what those slots cost in a healthy network: the $\wasync - \wsync$ rounds of one slot in \period, and the one round its successor waits for causal ordering, which our evaluation measures at about $1\%$ of latency at $\period = 64$ (\Cref{sec:eval-healthy}).

We keep the protocol as it is for three reasons. The \outputreading would no longer be a single rule with a slot-dependent wavelength, so the conservativity of \Cref{thm:conservativity} (\Cref{app:full-proofs}) would become a statement about a reading rather than about what the protocol outputs. The two readings would no longer consider the same slots, as the same round would have a known leader in one and a coin leader in the other, so the agreement argument could no longer reuse \Cref{cor:handover} (\Cref{app:full-proofs}) and would have to be made once per reading. Nothing would be saved beyond that latency, since the coin is opened for the \controlreading either way. We prefer the smaller proof to saving a percent of latency, particularly because \period is large in any deployment.

\section{Full Proofs}
\label{app:full-proofs}

This appendix states every result of \Cref{sec:proofs} and proves it. Each proof names the clause of \Cref{sec:interface} at the step that relies on it; every argument that needs neither \cref{clause:a4} nor \cref{clause:a5} holds for any pair of rules of the wave family, at any wavelength function. \Cref{thm:agreement} also holds beyond the wave family, for any family of rules that satisfy \cref{clause:a2,clause:a3} and the uniqueness of \Cref{lem:cert-unique}, and that agree on the tie-break and on the number of links the indirect step tries (\Cref{app:lean-model}): composed so that each slot is decided by the rule of its kind, they form one rule with the same properties, and \sysname is that composite for $\{\rsync, \rasync\}$. The Lean development of \Cref{app:lean} machine-checks these statements, and the arguments here follow the machine-checked ones.

Correctness rests on the interface of \Cref{sec:model}: a schedule of wavelengths $\wl : \mathbb{N} \to \{\wsync, \wasync\}$ with $\wl(r) \ge 2$, and two rules satisfying \crefrange{clause:a1}{clause:a5}. Statements apply to local views of one universe of blocks, and the properties established are those of \Cref{def:bab}. The two lemmas below rely only on \cref{clause:a1}, that is, on a shared committee with quorums of $n - f$ and one block per honest author and round, and on two rules of \Cref{sec:model-dag}: every block references $n - f$ blocks of the round below it, and a block votes for the first block it sees from the leader; they discharge \cref{clause:a3} (\Cref{lem:cert-unique}) and \cref{clause:a2} (\Cref{lem:intersection}) for both pairs; for the $3f + 1$ pair, \Cref{lem:cert-unique} is \mahimahi's certificate lemma read at the slot's own wavelength. They read a rule through two parameters. Its \emph{support} for a candidate of the slot at round $r$ is the set of authors whose block at the decision round $r + \wlof{r} - 1$ carries the rule's evidence for it: for the $3f + 1$ pair a certificate, a block referencing $n - f$ votes for the candidate from round $r + \wlof{r} - 2$, at every $\wl \ge 3$; for BlueBottle's pair a vote for the candidate, cast directly at $r + 1$ by the synchronous variant and through the causal history at $r + 2$ by the asynchronous one. Its \emph{indirect threshold} is the number of supporters an anchor's causal history must hold to commit the slot indirectly, one for the $3f + 1$ pair and $n - 3f$ for BlueBottle's. Under every rule a direct commit needs $n - f$ supporters, and a direct skip $n - f$ \emph{blames}, blocks of one round of the wave that vote for no candidate, the vote round for the $3f + 1$ pair and the decision round for BlueBottle's. \Cref{tab:discharge}, at the end of this appendix, maps each clause to each base protocol.

\begin{lemma}[Support uniqueness]\label{lem:cert-unique}
    For the $3f + 1$ pair, at most one candidate of a slot is ever certified, and none if the slot is directly skipped. For BlueBottle's pair, a candidate with $n - f$ supporters leaves every rival below $n - 3f$ supporters, and a directly skipped slot leaves every candidate below $n - 3f$.
\end{lemma}
\begin{proof}
    An honest author has one block per round (\cref{clause:a1}), and that block votes for at most one candidate of a slot, the first block it sees from the leader (\Cref{sec:model-dag}). Two sets of authors whose blocks at one round vote for two different candidates, or one for a candidate and one for none, therefore share only Byzantine authors and together number at most $n + f$; if one set has $n - f$ authors, the other has at most $2f$.

    \emph{The $3f + 1$ pair.} Count at the vote round. A certified candidate has $n - f$ votes there, since a certificate references $n - f$ votes, and a directly skipped slot has $n - f$ blames there. A second certified candidate, or a certified candidate of a skipped slot, would need $n - f$ votes against the $2f$ that are left, and $n - f > 2f$ at $n \ge 3f + 1$.

    \emph{BlueBottle's pair.} Count at the decision round, where the supporters are the voters themselves. Against a candidate's $n - f$ supporters every rival keeps at most $2f$ supporters, and against the $n - f$ blames of a direct skip every candidate does; $2f < n - 3f$ because $n \ge 5f + 1$. \qed
\end{proof}

\begin{lemma}[Quorum intersection across the wave]\label{lem:intersection}
    If a candidate of the slot at round $r$ has $n - f$ supporters, the causal history of every block at round $r + \wlof{r}$ or above reaches the slot's indirect threshold for it: it holds a certificate for the candidate under the $3f + 1$ pair and the supporting blocks of at least $n - 3f$ of its supporters under BlueBottle's.
\end{lemma}
\begin{proof}
    A block $B$ at round $r + \wlof{r}$ references $n - f$ blocks of the decision round $r + \wlof{r} - 1$ (\Cref{sec:model-dag}). Their authors and the $n - f$ supporters both lie among the $n$ authors of one round, so at least $n - 2f$ authors are in both sets, and at least $n - 3f$ of those are honest. An honest author has one block at the round (\cref{clause:a1}), its supporting block, so $B$ references it outright; for a Byzantine author in the intersection the block $B$ references may be a twin carrying no support, which is why the count stops at the honest ones. At $n \ge 3f + 1$ this is at least one certificate, and under BlueBottle's pair it is the $n - 3f$ supporters its indirect step asks. A block above round $r + \wlof{r}$ references a block of the round below it (\Cref{sec:model-dag}), so by induction on its round it reaches a block at round $r + \wlof{r}$, and causal history is transitive. \qed
\end{proof}

\begin{theorem}[Agreement]\label{thm:agreement}
    For every slot, no two honest validators hold conflicting verdicts, regardless of their local views and whether they decided directly or indirectly. At the adaptive schedule this holds for validators that have each derived the period of every interval up to the highest round of the (finite) universe of blocks.
\end{theorem}
\begin{proof}
    Fix a wavelength function \wl and the schedule it induces; the last paragraph extends the claim to validators running their own derived periods. Let $s$ be the slot at round $r$ and let two honest validators $u$ and $v$ hold verdicts for it. The proof is by induction on the derivation of $u$'s verdict: the indirect step derives a verdict for $s$ from verdicts for slots strictly above it, and the induction hypothesis is that each of those agrees with every verdict $v$ holds for the same slot.

    \emph{Two direct verdicts.} Two direct commits name the same block, and a direct commit and a direct skip cannot coexist, both by \Cref{lem:cert-unique}.

    \emph{One direct and one indirect verdict.} This case uses no induction hypothesis, so let $v$ be the validator that decides directly, and $u$ decide $s$ through its anchor $A$. If $v$ directly commits $s$, its candidate has $n - f$ supporters at the decision round $r + \wlof{r} - 1$. The anchor search starts at round $r + \wlof{r}$ (\Cref{sec:interleave-safety}), so $A$ lies at that round or above, whichever rule decides $A$ and whatever wavelength $A$'s own slot carries. By \Cref{lem:intersection} (\cref{clause:a2}) the causal history of $A$'s leader block reaches the slot's indirect threshold for $v$'s candidate, and by \Cref{lem:cert-unique} no rival does, so $u$ commits the same candidate. If instead $v$ directly skips $s$, \Cref{lem:cert-unique} (\cref{clause:a3}) leaves every candidate below the indirect threshold, so $u$'s anchor holds none at it and $u$ skips as well.

    \emph{Two indirect verdicts.} Both searches start at round $r + \wlof{r}$, a function of $r$ and of the schedule alone, and both anchors are commits, since an \pundecided anchor yields no verdict. A slot $u$'s search passes over is a skip in $u$'s derivation, which by the induction hypothesis $v$ does not commit, so $v$'s anchor does not lie below $u$'s; and $u$'s anchor is a commit in $u$'s derivation, which by the induction hypothesis $v$ does not skip, so $v$'s search does not pass over it. Both validators therefore land on the same slot $A$, where by the induction hypothesis they commit the same block. A block determines its own causal history, so both examine the same supporting blocks, apply the same tie-break, and reach the same verdict.

    \emph{At the adaptive schedule.} Let $u$ and $v$ each run their own derived periods, and let each have derived the period of every interval up to the highest round of the universe of blocks. By \Cref{thm:agreement-period} the two sequences of periods agree on those intervals, so every slot either derivation reads has the same wavelength at both, and the argument above applies on one schedule. There is no circularity: the proof of \Cref{thm:agreement-period} invokes the present theorem only at the \controlslots of one scan, whose schedule is fixed by the interval's period, the period bound and the interval boundary. \qed
\end{proof}

\begin{corollary}[Handover]\label{cor:handover}
    If some honest validator directly commits the slot at round $r$, every honest validator whose anchor for that slot is committed commits it too, whichever rule decides the anchor, and no honest validator ever skips it.
\end{corollary}
\begin{proof}
    An honest validator whose anchor for the slot is committed holds a verdict for the slot, its direct one or else the indirect step's, which returns a commit or a skip, and by \Cref{thm:agreement} that verdict agrees with the direct commit; the same theorem excludes a skip at every honest validator. \qed
\end{proof}

\para{Where the two rules meet}
The case of one direct and one indirect verdict in the proof of \Cref{thm:agreement}, whose commit half \Cref{cor:handover} states, is the only step where the two rules interact: an anchor committed by one rule completes a direct commit of the other. Two consequences are used later. A period change never has to pause the protocol to finish pending slots: the anchors of the next period decide the slots left \pundecided under the previous one. And the \outputreading and the \controlreading of \Cref{sec:adaptive-control} read a direct verdict identically, so their verdicts on a shared slot can differ only where both are indirect.

\para{Why the floor is the slot's own wavelength}
\Cref{lem:intersection} reads the wavelength of the slot, $\wlof{r}$, not that of the slot at the round of the block whose causal history it reads; this is what \cref{clause:a2} asks of a rule, and it sets the anchor floor at $r + \wlof{r}$. Neither $r + 1$ nor $r + \wl(\text{anchor})$ will do. Consider an \asyncslot $L$ at round $r$ with $\wasync = 5$: its certificates sit at round $r + 4$, and a validator holding $n - f$ of them commits $L$ directly. An anchor below round $r + \wasync - 1$ has its entire causal history below the round those certificates sit at, so it holds none of them and its indirect step skips $L$, contradicting the direct commit; an anchor at round $r + \wasync - 1$ itself is one block of that round, which need not be among the certifiers. Only from round $r + \wasync$ on does \Cref{lem:intersection} guarantee a certificate. A \syncslot $P$ at round $r + 1$ is the extreme case, its blocks referencing only round-$r$ blocks, but a \syncslot at round $r + \wsync$ is no better. \Cref{fig:anchor} draws both verdicts on one \pdag: the rule's floor commits the slot through the \syncslot at round $9$, while a floor read from the anchor's own wavelength lands on the \syncslot at round $7$ and skips it.

\begin{figure}[t]
    \centering
    {\centering\input{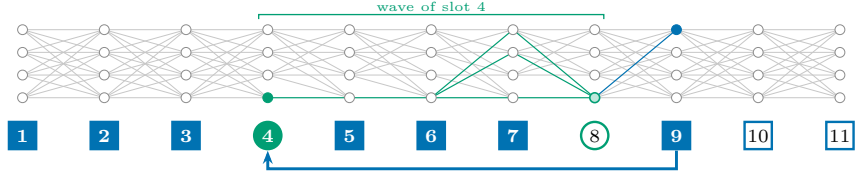}\par}
    \subcap{The floor of \Cref{sec:interleave-safety}, $4 + \wasync = 9$. The search lands on slot $9$, and the edge from its leader block to the certificate at round $8$ is why slot $4$ commits.}\label{fig:anchor-rule}
    {\centering\input{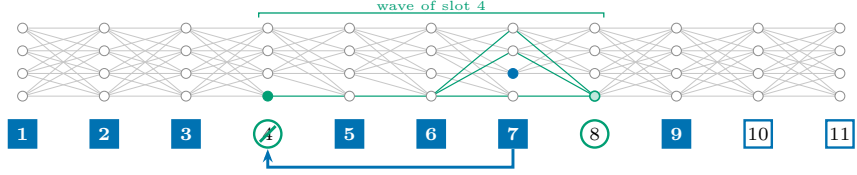}\par}
    \subcap{A floor read from the anchor's own wavelength, $4 + \wsync = 7$. The search lands on slot $7$, whose causal history lies entirely below round $8$, so no edge reaches the certificate and slot $4$ is skipped.}\label{fig:anchor-low}
    \caption{One \pdag read with two anchor floors ($n = 4$, $\wsync = 3$, $\wasync = 5$, $\period = 4$; the \pdag of \Cref{fig:running}). Highlighted in the \pdag are the leader block of slot $4$ at round $4$, the one block that certifies it at round $8$, every route between the two, of which the three top edges are the round-$7$ votes that make that block a certificate, and the leader block of the anchor the search lands on. The routes funnel through the leader's own chain, since no other author's block at rounds $5$ or $6$ references it, which is why slot $4$ has exactly one certificate. The bracket marks the wave of slot $4$, and the arrow runs from the anchor to the slot it decides. In the strip below each panel, squares are \syncslots and circles \asyncslots, filled on a commit and struck on a skip. One view yields both verdicts. No validator commits slot $4$ directly here, but the low floor can skip a slot that another validator has committed directly, which the floor of the rule excludes (\Cref{lem:intersection}).}
    \label{fig:anchor}
\end{figure}

\begin{corollary}[Total order and integrity]\label{cor:order-integrity}
    Slots are output in round order, and each block is output at most once.
\end{corollary}
\begin{proof}
    \emph{Total order.} Output proceeds in slot order and stops at the first \pundecided slot (\Cref{sec:model-dag}), so what a validator has output is the sequence of commit verdicts over a settled prefix of slots, and settling further slots only appends. By \Cref{thm:agreement} two validators agree on the verdict of every slot both have settled, so their outputs agree as far as both reach and neither reorders what it already released.

    \emph{Integrity.} Both the direct and the indirect step commit only a candidate of the slot under decision, that is, a block of the slot's round authored by that slot's leader. \sysname places one slot at each round, so a committed leader block enters the output only at the slot of its own round, which carries a single verdict. Every other block enters with the first committed leader whose causal history holds it, and a later commit does not output it again. \qed
\end{proof}

\begin{theorem}[Liveness under partial synchrony]\label{thm:liveness-ps}
    After \pgst, with bounded leader timeouts, every \syncslot whose leader is honest and whose successor round carries a leader wait is decided as a commit by the direct rule within \wsync rounds, and every \asyncslot whose elected leader is honest (or crashed) is decided directly within \wasync rounds, and partial dissemination alone does not defer the decision of an \asyncslot. Any slot the direct rule leaves \pundecided, such as one whose leader equivocates, is decided by its anchor once some slot above its floor commits and every slot in between is decided.
\end{theorem}
\begin{proof}
    \emph{Synchronous slots.} After \pgst and under the pacing of \Cref{sec:interleave-schedule}, a \syncslot whose leader is honest and whose successor round carries a leader wait is directly committed within \wsync rounds. This is \cref{clause:a4}.

    \emph{Asynchronous slots with an honest or crashed leader.} After \pgst, under the reference condition of \Cref{app:lean-assumptions}, a block held by one honest validator is referenced, transitively, by every honest block two rounds later. An honest leader's block therefore lies in the causal history of every honest block from round $r + 2$ on, the honest voters vote for it, and every honest block at the certify round references all of them and so certifies it: the slot is directly committed within \wasync rounds, at $\wasync \ge 4$. BlueBottle's asynchronous variant needs no certify round: its voters two rounds up reach the leader's block through their causal history, so the slot commits at $\wasync = 3$. A crashed or silent leader collects $n - f$ blames at the vote round and is directly skipped. Partial dissemination therefore does not defer the decision: it suffices that one honest validator holds the leader block in time, and the leader may be Byzantine as long as it did not equivocate.

    \emph{The wait above an \pundecided slot.} What remains is a slot the direct rule leaves \pundecided, for instance one whose leader equivocates or a \syncslot whose successor round carries no leader wait, and its wait is set by the leader schedule rather than by the decision rule. Consider the \emph{floor chain} of the slot: hop from the slot to the first slot at or above its floor that the validator does not skip, then hop again from there. Crashed and silent leaders are skipped, so every landing is a slot whose leader did propose, and a landing the first two paragraphs cover, an honestly led \asyncslot or an honestly led \syncslot whose successor round carries a leader wait, commits directly. Going back down the chain, a landing whose own successor landing commits is decided by the indirect step and not skipped, hence commits in turn and anchors the landing below it. It is therefore enough that the chain reach a committed landing; if some slot above the floor commits and every slot in between is decided, the first landing is already a commit, which is the second claim. \qed
\end{proof}

\para{Round counts}
The wait above an \pundecided slot has a bound at the two ends of the dial. At $\period = \infty$, under a round-robin schedule with $\wsync\,f < n$, no two landings share a residue class of the rotation. Were two landings at one residue, they would bound a whole number of cycles; an honestly led round commits and so is never one of the skips a hop passes over, hence honest leaders could sit only in the $\wsync - 1$ rounds between a landing and its own floor, and counting the $f$ Byzantine residues against the cycles gives $n \le \wsync\,f$. The landings' leaders being distinct, the $b \le f$ Byzantine validators occupy at most $b$ of them and one of the first $b + 1$ landings is honestly led. An honestly led round lies within $f$ rounds above any floor, so a hop climbs at most $\wsync + f$ rounds, and every decision round the argument reads lies within $(b + 1)(\wsync + f)$ rounds of the slot.

At $\period = 1$ the wait is the first run of \wasync coin rounds naming directly committed candidates, since such a run decides every slot below it (the proof of \Cref{thm:liveness-async}). Cut the coin rounds into blocks of \wasync, and call a block good if every one of its coins names a directly committed candidate. Conditioning on the coins already revealed, \cref{clause:a5} gives each round probability at least $p$, so a block is good with probability at least $p^{\wasync}$, and the expected number of blocks up to and including the first good one is at most $1 / p^{\wasync}$. No bound per hop of the search holds under the coin: the coin that commits an anchor above a pending slot both skips that slot and moves the landing, so the landings are not independent draws, and a two-hop bound of $(b/n)^2$ fails on a concrete \pdag (\Cref{app:lean}). In between, the rotation leads only the \syncslots and the coin the rest, and the residue count needs slack in the committee that $n = 3f + 1$ does not have.

\begin{theorem}[Liveness under asynchrony]\label{thm:liveness-async}
    Assume $\pmaxperiod \le \pinterval$, $\wasync \le \pinterval$ and $\wsync \le \wasync$, that every round above $2\,\pinterval$ is populated (\cref{clause:a1}), and that the counting bound of \cref{clause:a5} holds at every round with $p > 0$. Under asynchrony, with probability $1$, above every interval some interval finds a pivot (\Cref{sec:adaptive-control}), and every slot is eventually decided. While a slot is \pundecided, every pivot at least two intervals above it hands the next interval $\period = 1$ through the failover.
\end{theorem}
\begin{proof}
    \emph{Good blocks.} Fix a slot $s$ and take blocks of $\wasync \cdot \pmaxperiod$ consecutive rounds, each opening an interval: the first at the second interval above that of $s$, and one more every $\lceil \wasync \cdot \pmaxperiod / \pinterval \rceil$-th interval, so that no two overlap. Read the coin as a draw at every round, of which the protocol reveals those at its \asyncslots. A block is \emph{good} if the coin names a directly committed candidate at each of its rounds. Conditioning on the coins already revealed, \cref{clause:a5} gives each round probability at least $p$, so a block is good with probability at least $p^{\wasync \cdot \pmaxperiod}$ whatever the earlier blocks did.

    \emph{Some interval finds a pivot.} A good block holds \wasync consecutive multiples of \pmaxperiod, which every scan of an earlier interval reads as consecutive \controlslots above its boundary (\Cref{sec:adaptive-control}). They commit directly and, as the run of the next paragraph does at the \controlreading's fixed wavelength \wasync, decide every \controlslot below them, so every earlier scan closes and the period of the interval the block opens is derived. That period is at most $\pmaxperiod \le \pinterval$, so the interval's first \controlslot lies among the block's first \pmaxperiod rounds; it commits directly and is the interval's pivot. Good blocks occur with probability $1$, by the bound at the end of this proof.

    \emph{Every slot is eventually decided.} Let a good block give an interval at least two past that of $s$ its pivot, and let a later good block, which settles every scan below its own interval, supply in its first \wasync rounds, which lie in that interval since $\wasync \le \pinterval$, a run of coins naming directly committed candidates. If $s$ is still \pundecided, the agreed output's last commit lies below $s$ and the pivot more than \pinterval rounds above $s$, so the failover hands the next interval $\period = 1$, and every later interval up to the run keeps it while $s$ stays \pundecided. The run's slots are then \asyncslots that commit directly, and since no slot's wavelength exceeds \wasync, the run decides every slot below it: the highest \pundecided slot below the run has its floor at or below the run's last slot, the first slot at or above that floor that is not a skip is a commit, and the indirect step returns a verdict. Two good blocks, one in each half of $M$ blocks, occur except with probability at most $2(1 - p^{\wasync \cdot \pmaxperiod})^{\lfloor M/2 \rfloor}$, which vanishes. The slots being countably many, with probability $1$ every slot is eventually decided, with the run's commits above it, and by \Cref{cor:order-integrity} the ledger grows with them. \qed
\end{proof}

\para{Period updates while output is stalled}
A direct commit reaches a lower slot only through a decided stretch: from the lower slot's floor up to the first commit, every slot must be decided, because the anchor search stops at an \pundecided one. Suppose no \syncslot is ever certified or directly skipped, as an adversary can arrange under asynchrony (\Cref{sec:interleave-limits}), so that no \syncslot commits. At $\period \ge \wsync$ the stretch above a slot at a round $r$ with $r \bmod \period = \period - 1$ then ends at an \asyncslot above round $r + \period$ and so holds the slot at round $r + \period$, which is in the same position; no such slot is therefore ever decided, since each needs the one a period up decided first. Take $\period = 4$, $\wsync = 3$, and $\wasync = 5$. The \syncslot at round $3$ searches from round $6$, and since no \syncslot commits, only a committed \asyncslot at round $8$ or above can decide it; the search reaches round $8$ only if slots $6$ and $7$ are skipped (in \Cref{fig:readings} it waits at slot $6$), and slot $7$ is in the same position one period up, and so on. Meanwhile the \asyncslots at $8$ and $12$ commit directly and are never output. Because a directly committed \asyncslot is read identically in both readings, the two readings differ only when an \asyncslot is not directly committed. Suppose the direct rule leaves the \asyncslot at round $4$ \pundecided. In the \outputreading, its anchor search begins at round $9$ and, since no \syncslot commits, only a committed \asyncslot at round $12$ or above can decide it, which needs slot $11$ skipped; slot $11$ is never decided, so slot $4$ stays \pundecided. In the \controlreading, whose slots above round $8$ are the multiples of $\pmaxperiod = 4$ as in \Cref{fig:readings}, the search visits only \controlslots and reaches slot $12$, whose causal history decides it. The pivot of an interval, the slot whose causal history determines the next period, must be identical for validators that observed a direct commit and those that did not. Only the \controlreading lets the latter resolve it, which is why the period update evaluates \rasync at the \controlslots rather than at the output.

The counting lemma underpinning the asynchronous bound requires $\wasync \ge 4$: at $\wasync = 4$ it guarantees only one directly committed candidate, so a commit has probability at least $1/n$, and at $\wasync = 5$ the stronger $n - f - b$ bound holds (both derived under ``Instantiating the theorems'' below). The $3f + 1$ pair inherits this trade-off from \mahimahi; BlueBottle's pair does not, its counting lemma holding at $\wasync = 3$. The bound also shows what the sparse \controlslots cost. An interval contains $c = \pinterval / \period$ \controlslots, with at least two slots when $\period = \pmaxperiod$. Since an interval retains its period when all of its \controlslots are skipped, counting direct commits alone, a scan finds its pivot within an expected $1 / (1 - (1 - p)^{c})$ intervals, which is approximately $1.8$ intervals for $p = 1/3$ and $c = 2$, but approximately $n/2$ intervals for the $p = 1/n$ of $\wasync = 4$. A deployment at $\wasync = 4$ should therefore keep the ratio $\pinterval / \pmaxperiod$ well above two.

\begin{theorem}[Agreement of the period]\label{thm:agreement-period}
    Under any deterministic update rule, two honest validators that derive the period of interval $j$ derive the same one.
\end{theorem}
\begin{proof}
    By induction on configurations. Configuration $j$ is interval $j$, the rounds $j\,\pinterval + 1$ through $(j + 1)\,\pinterval$ (\Cref{sec:adaptive}), and its state is the period $\period_j$ that the scan of interval $j - 1$ fixed, the agreed output's next slot, and the round of its last committed leader. Interval $0$ runs at the initial period. Assume that any two honest validators holding a state for interval $j$ hold the same one.

    \emph{The scan's slots agree.} The \controlslots of the scan of interval $j$ are the multiples of $\period_j$ up to the boundary $(j + 1)\,\pinterval$ and the multiples of \pmaxperiod above it (\Cref{sec:adaptive-control}). They are a function of $\period_j$, of \pmaxperiod and of the boundary, all agreed, and of nothing that a view holds.

    \emph{The scan's verdicts agree.} The \controlreading is \rasync on the sub-schedule those slots name, at the fixed wavelength \wasync. \Cref{thm:agreement} applies to it unchanged, its proof reading only \Cref{lem:cert-unique,lem:intersection} and a common anchor search, none of which mentions the output's slots. Each scan has its own \controlslots, and verdicts of different scans are never compared.

    \emph{The pivot agrees.} The pivot is the first \controlslot of interval $j$ whose \controlverdict is a commit, every earlier one being a skip. Both components are agreed by the previous step, and so is the case of no pivot, in which the state carries over unchanged.

    \emph{The next state agrees.} The agreed output is the output rule applied to the pivot's causal history, continued from the state's cursor. A block determines its own causal history, so both validators read the same blocks and, the output rule being deterministic, derive the same verdicts, and they move the cursor and the last commit alike. Every verdict that advance reads lies at or below the pivot's round, where the periods are agreed, so the induction closes at each interval. The failover then compares the round of the agreed output's last commit with the pivot's round, both agreed, and the update rule is a deterministic function of the pivot block, the window it determines, and the current period. \qed
\end{proof}

\begin{theorem}[Conservativity]\label{thm:conservativity}
    With a constant rule that never changes \period, setting $\period = 1$ yields exactly the verdicts of \rasync and setting $\period = \infty$ exactly those of \rsync (\mahimahi's and \mysticeti's, in the $3f+1$ pair), on a given \pdag.
\end{theorem}
\begin{proof}
    At $\period = 1$ every round satisfies $r \bmod \period = 0$, so $\wlof{r} = \wasync$ everywhere, every slot is an \asyncslot with a coin-elected leader, and both the direct and the indirect step are \rasync's at \wasync. At $\period = \infty$ no round satisfies it, so $\wlof{r} = \wsync$ everywhere and every slot is a \syncslot. In both cases the anchor floor $r + \wlof{r}$ collapses to the constant floor of the base rule, so the anchor searches coincide as well, and the verdicts agree slot by slot. \qed
\end{proof}

\begin{lemma}[The replay's direct-commit weight]\label{lem:replay-starvation}
    For the $3f + 1$ pair at $\wasync = 5$ and for BlueBottle's pair, let $r$ be a round of a window that holds the blocks of a quorum, honest for BlueBottle's pair, at the rounds the counting lemma reads for $r$: the vote and decision rounds of $r$'s wave for \mahimahi, the two rounds above $r$ for BlueBottle's asynchronous variant. Then $c_r / n$ is at least the counting fraction of \rasync's \cref{clause:a5} under any message scheduling ($c_r \ge n - f - b$ for \mahimahi).
\end{lemma}
\begin{proof}
    Read the window as a record of its own: its support is that of the universe restricted to it, and by hypothesis a quorum populates within it the rounds the counting lemma reads for $r$. The counting property of \cref{clause:a5} then applies to every candidate at once and names at least $pn$ authors whose blocks carry direct-commit evidence inside the window, which for \mahimahi at $\wasync = 5$ is $c_r \ge n - f - b$ under any scheduling. By construction $c_r / n$ is the share of the $n$ candidate leaders that the replay marks committed at round $r$, which is the probability that a uniform coin names a directly committed leader there. \qed
\end{proof}

\begin{corollary}[Atomic broadcast]\label{cor:bab}
    \sysname implements atomic broadcast (\Cref{def:bab}) for any pair of rules satisfying \crefrange{clause:a1}{clause:a5}.
\end{corollary}
\begin{proof}
    \emph{Agreement, total order, and integrity.} By \Cref{thm:agreement,cor:order-integrity}, honest validators agree on every slot both have settled, output settled slots in order, and output each block at one slot only, so a block one of them delivers lies in the output of every other whose settled prefix reaches as far, and with probability $1$ every honest validator's settled prefix eventually reaches as far (\Cref{thm:liveness-ps,thm:liveness-async}). Every output block is a block of the \pdag and so carries its author's signature (\Cref{sec:model-system}).

    \emph{Validity.} An honest block lies in the causal history of every honest block from some later round on: after \pgst the round above its own, and under asynchrony the round by which the reference rule of \cref{clause:a1} has spread it to every honest validator. Every block above that round, whatever its author, references $n - f$ blocks of the round below (\Cref{sec:model-dag}), at least one of them honest, and so holds it too. The first slot committed above that round therefore delivers it, whichever leader it has, once every slot below it is decided. With probability $1$ every slot is decided (\Cref{thm:liveness-ps,thm:liveness-async}), and a slot above that round commits: after \pgst an honestly led one (\Cref{thm:liveness-ps}), and under asynchrony a slot of the run in the proof of \Cref{thm:liveness-async}. \qed
\end{proof}

\para{Parameters and guards}
The loop of \Cref{sec:adaptive-loop} defers six details here. None of them enters a safety argument, and \Cref{thm:agreement-period} holds whatever they are set to, since each is a function of round numbers and of agreed state.

\emph{The gating rule.} A slot is evaluated once its interval's period is known. This needs no separate mechanism: a validator that has not closed the scan of interval $j$ derives no state for interval $j + 1$, so it has no wavelength for those rounds and decides nothing there. \emph{The one-interval lag.} A period computed from the scan of interval $j$ applies from interval $j + 1$ on, never to interval $j$ itself. Without the lag the period of a round would be needed in order to decide the very slots whose verdicts fix it.

\emph{The bounds on \pinterval.} Two are needed, and the first does not imply the second. $\pinterval \ge 2\,\pmaxperiod$ makes every interval, and every window of a pivot at round \pinterval or above, hold at least two \asyncslots of every candidate period. $\pinterval \ge \pmaxperiod + \wasync - 2$ makes a window of $\pinterval + 1$ rounds hold the \emph{decision} round of at least one \asyncslot of every candidate; under the first bound alone, a window resolves such a slot at some pivots and none at others, and the replay then scores a candidate on no evidence. The garbage-collection horizon must retain at least the last $\pinterval + 1$ rounds, which is what a window reads.

\emph{The warm-up interval.} Interval $0$ hands interval $1$ its own period whatever the replay answers, because its window holds the start-up rounds, where no rule has had a full wave to commit anything. The failover cannot fire there either: the pivot lies at round \pinterval or below and the agreed output's last commit at round $0$ or above, so the test $\ell + \pinterval < \mathrm{round}(A)$ of \Cref{alg:update} fails.

\emph{Hysteresis and ties.} The selection leaves the current period only for a candidate scoring below $(1 - \epsilon)$ times the current period's, and when it does leave, it takes the best candidate. The best candidate scores no worse than the current period and no worse than any candidate, is the current period whenever that one scores as well, and is otherwise the largest candidate attaining the best score, in whatever order the candidates are listed. \emph{The candidate set.} $K = \{1, 2, 4, \dots, \pmaxperiod\}$, the powers of two up to the bound. Every candidate divides \pmaxperiod, which is what makes the multiples of \pmaxperiod \asyncslots under every candidate period, hence defined and coin-carrying whatever the scan decides (\Cref{sec:adaptive-control}), and what keeps every derived period a divisor of the bound. Because \pcanary is odd it is coprime to every candidate, so a window holding two canary rounds whose decision round it retains holds a probe for every candidate of at least two: two consecutive multiples of the canary spacing cannot both be multiples of the period.

\begin{algorithm}[t]
    \caption{Period update on the committed \pdag}
    \label{alg:update}
    \begin{algorithmic}[1]
        \Statex \textit{The interval scan reads \rasync at the \controlslots (\Cref{sec:adaptive-control}): the \asyncslots of interval $j$ and, above its boundary, the multiples of \pmaxperiod. Evidence is the support of \Cref{lem:cert-unique}; $\beta(\wl)$ is the offset of the rule's blame round, $\wl - 2$ for the $3f + 1$ pair and $\wl - 1$ for BlueBottle's.}
        \Procedure{\pupdateperiod}{$j$} \Comment{scan interval $j$'s \controlslots upward; wait at undecided}
        \State $A \gets$ the first \controlslot with a commit verdict in the scan \Comment{earlier slots are skipped}
        \State \textbf{if} there is none: \Return \Comment{no commit keeps \period}
        \State $W \gets$ \Call{\pgetsubdag}{$A$, \pinterval} \Comment{causal history of $A$ at rounds $\mathrm{round}(A) - \pinterval$ and above}
        \State extend the agreed output from its last slot: \Comment{identical at every validator}
        \State \quad \ptrycommit on the causal history of $A$
        \State $\ell \gets$ round of the agreed output's last commit, or $0$ if none
        \State \textbf{if} $\ell + \pinterval < \mathrm{round}(A)$: $\period \gets 1$; \Return \Comment{failover; startup grace through round $\pinterval$}
        \State \textbf{for} $\period' \in K$: $L[\period'] \gets$ \Call{\preplay}{$W$, $\period'$} \Comment{expected output delay under period $\period'$}
        \State $\period^\star \gets \arg\min_{\period' \in K} L[\period']$ \Comment{ties keep \period, then favor the larger candidate}
        \State \textbf{if} $L[\period^\star] < (1 - \epsilon) \cdot L[\period]$: $\period \gets \period^\star$ \Comment{hysteresis; applies above round $(j+1)\,\pinterval$}
        \EndProcedure
        \Statex
        \Procedure{\pproberate}{$W$, $\period'$} \Comment{canary rounds that $\period'$ replays as \syncslots}
        \State $P \gets \{ r \in W : r \bmod \period' \neq 0,\ r \bmod \pcanary = 0,\ r + \wsync - 1 \le \mathrm{top}(W) \}$
        \State $s \gets |\{ r \in P : W \text{ holds } n{-}f \text{ supporters of } \Call{\pgetleader}{$r$} \text{ at } r + \wsync - 1 \}|$
        \State \Return $(s, |P|)$ \Comment{successes, probes}
        \EndProcedure
        \Statex
        \Procedure{\preplay}{$W$, $\period'$} \Comment{sum over $W$ of the expected delay from a round to its output}
        \State $(s, t) \gets$ \Call{\pproberate}{$W$, $\period'$}
        \State \textbf{for} each round $r$ of $W$, from the top downward: \Comment{pass 1: decide every slot under $\period'$}
        \State \quad $\wl \gets \wasync$ \textbf{if} $r \bmod \period' = 0$ \textbf{else} $\wsync$
        \State \quad \textbf{if} $r + \wl - 1 > \mathrm{top}(W)$: $\mathit{dec}_r \gets \mathit{com}_r \gets \mathrm{top}(W)$; \textbf{continue} \Comment{same penalty for every $\period'$}
        \State \quad \textbf{if} $\wl = \wsync$ \textbf{and} $r \bmod \pcanary \neq 0$ \textbf{and} $t > 0$: \Comment{unprobed \syncslot}
        \State \quad \quad $\mathit{dec}_r \gets \big(s\,(r + \wsync - 1) + (t - s)(r + \beta(\wsync))\big)/t$
        \State \quad \quad $\mathit{com}_r \gets \big(s\,(r + \wsync - 1) + (t - s)\,\mathrm{top}(W)\big)/t$; \textbf{continue} \Comment{commits w.p.\ $s/t$, else skips}
        \State \quad $a \gets$ \Call{\pfindanchor}{$r$, $\wl$} \Comment{earliest slot at round $\ge r + \wl$ with $\mathit{com}_a < \mathrm{top}(W)$}
        \State \quad \textbf{if} $\wl = \wsync$: $\mathit{cand} \gets \{\Call{\pgetleader}{$r$}\}$ \Comment{known leader: one candidate, exact}
        \State \quad \textbf{else}: $\mathit{cand} \gets$ all $n$ authors \Comment{hidden leader: every author is a candidate, weight $1/n$ each}
        \State \quad \textbf{for} each $v \in \mathit{cand}$, with $B$ the block of $v$ at round $r$ in $W$ (possibly none):
        \State \quad \quad \textbf{if} $n{-}f$ blocks at $r + \beta(\wl)$ blame $B$: $\mathit{dec}_v \gets r + \beta(\wl)$; $\mathit{com}_v \gets \mathrm{top}(W)$ \Comment{direct skip}
        \State \quad \quad \textbf{else if} $n{-}f$ blocks at $r + \wl - 1$ support $B$: $\mathit{dec}_v \gets \mathit{com}_v \gets r + \wl - 1$ \Comment{direct commit}
        \State \quad \quad \textbf{else}: $\mathit{dec}_v \gets \mathit{dec}_a$ \Comment{indirect, by the anchor}
        \State \quad \quad \phantom{\textbf{else}: }$\mathit{com}_v \gets \mathit{com}_a$ \textbf{if} $W$ holds the indirect threshold of supporters of $B$ \textbf{else} $\mathrm{top}(W)$
        \State \quad $\mathit{dec}_r \gets \mathrm{mean}_{v \in \mathit{cand}}\, \mathit{dec}_v$;\quad $\mathit{com}_r \gets \mathrm{mean}_{v \in \mathit{cand}}\, \mathit{com}_v$ \Comment{expectation under a uniform coin}
        \State \textbf{for} each round $r$ of $W$, from the top downward: \Comment{pass 2: output at the first commit $\ge r$}
        \State \quad $C_r \gets \min(\mathit{com}_r, C_{r+1})$
        \State \textbf{for} each round $r$ of $W$, from the bottom upward: \Comment{pass 3: wait for every lower decision}
        \State \quad $G_r \gets \max(G_{r-1}, \mathit{dec}_{r-1})$
        \State \Return $\sum_r \big(\max(C_r, G_r) - r\big)$
        \EndProcedure
    \end{algorithmic}
\end{algorithm}

\para{The replay, in full}
\Cref{alg:update} states the update as the implementation runs it. The scan walks the \controlslots of the interval upward, waits at an \pundecided \controlverdict, takes the first commit as the pivot $A$, and keeps the period if every \controlslot is a skip. The \controlslots are fixed by rule and not read off the coin shares a validator holds: sets read from two views can differ, and two validators can then take their pivot at different rounds of one \pdag (\Cref{app:lean}). The scan then extends the agreed output over $A$'s causal history and applies the failover before consulting the score, so that a stalled output always gets $\period = 1$, whatever the score. \preplay scores a candidate $\period'$ by the expected delay from a round of the window to the output of its blocks, in three passes: the first decides every slot under $\period'$, exactly for a probed \syncslot, at the probes' success rate for an unprobed one, and averaged over the $n$ candidate leaders for an \asyncslot; the second gives each round the first commit at or above it; the third gates that on the decision of every lower slot, which accounts for head-of-line blocking. The cost is $O(|K| \cdot |W|)$ for $|K|$ candidates over a window of $|W|$ blocks.

Two approximations remain, both deterministic and identical for every candidate. Where the rule would decide a slot from the causal history of the anchor that the candidate would have used, the replay asks only whether the slot's indirect threshold of supporters lies anywhere in the window, the window being the pivot's own causal history and so standing in for the anchor's. And a slot that a candidate never commits inside the window is charged the window's top round, so the rounds it holds back are deferred there rather than left undefined. At the wavelengths of \Cref{lem:replay-starvation}, the first pass cannot be starved on the asynchronous side, provided the window holds a quorum's blocks, honest for BlueBottle's pair, at the rounds the counting lemma reads: the counting property counts support on the \pdag, while the replay reads the window. On the synchronous side the probes carry the estimate. A probe succeeds only on a support quorum the \pdag actually holds (\Cref{sec:adaptive-probes}), so the adversary can suppress evidence of the partially synchronous rule but never manufacture it; what it can do is serve the canary rounds' leaders alone, which lifts the rate the replay extends to the unprobed \syncslots above what those slots would have shown. And \Cref{lem:replay-starvation} bounds a rate, not a latency: it says how often a uniformly chosen leader would have committed directly, not that the replay selects the fastest period. The failover of \Cref{sec:adaptive-loop} covers that gap, which is why a window in which only the canaries were served goes to the failover and not to the score.

\para{Instantiating the theorems}
\Cref{tab:discharge} maps each interface clause to the result that discharges it for each of the four base protocols, \cref{clause:a4} being Corollary~2 of BlueBottle for its synchronous variant; the \controlverdicts rely on \cref{clause:a5} at the \controlslots. The counting fraction $p$ of \cref{clause:a5} is bounded per rule. For \mahimahi at $\wasync = 5$ the counting lemma (Lemmas~12 and~13 of that paper) gives $n - f$ directly committable blocks per populated round at $n = 3f + 1$; counting only the $n - f - b$ honest ones among them bounds $p$ below by $(n - f - b)/n$, at least $1/3$ at $n \ge 3f + 1$; at $\wasync = 4$ it promises only one committed candidate per populated round (its Lemma~15), so $p \ge 1/n$. BlueBottle's asynchronous variant counts at $\wasync = 3$ (Lemmas~28 to~31 of that paper): on a quorum of honest validators populating the two rounds above a round, at least $n - 3f$ honest blocks of that round are directly committed, so $p \ge (n - 3f)/n$.

\begin{table}[t]
    \centering
    \caption{Discharge table mapping interface clauses to their discharging lemma or base-paper result. In Lean (\Cref{app:lean}), \cref{clause:a1} is the model, except for block production and the reference rule, which the liveness results and validity take as hypotheses (\Cref{app:lean-assumptions}); \cref{clause:a2,clause:a3} are proved for both pairs; and \cref{clause:a4} and the counting floor of \cref{clause:a5} are proved on a given \pdag for the base protocols whose cells in those rows cite a base paper, and enter the generic results as hypotheses.}
    \label{tab:discharge}
    \scriptsize
    \begin{tabular*}{\textwidth}{@{\extracolsep{\fill}}lcccc@{}}
        \toprule
        Clause                           & \mysticeti                              & \mahimahi                               & BlueBottle (sync.)           & BlueBottle (async.)          \\
                                         & ($\wsync = 3$)                          & ($\wasync \in \{4, 5\}$)                & ($\wsync = 2$)               & ($\wasync = 3$)              \\
                                         & ($n \ge 3f + 1$)                        & ($n \ge 3f + 1$)                        & ($n \ge 5f + 1$)             & ($n \ge 5f + 1$)             \\
        \midrule
        \ref{clause:a1} (substrate)      & \multicolumn{2}{c}{shared $3f{+}1$ DAG} & \multicolumn{2}{c}{shared $5f{+}1$ DAG}                                                               \\
        \ref{clause:a2} (waved evidence) & \Cref{lem:intersection}                 & \Cref{lem:intersection}                 & \Cref{lem:intersection}      & \Cref{lem:intersection}      \\
        \ref{clause:a3} (skips exclude)  & \Cref{lem:cert-unique}                  & \Cref{lem:cert-unique}                  & \Cref{lem:cert-unique}       & \Cref{lem:cert-unique}       \\
        \ref{clause:a4} (synchronous)    & base paper~\cite{mysticeti}             & --                                      & base paper~\cite{bluebottle} & --                           \\
        \ref{clause:a5} (asynchronous)   & --                                      & base paper~\cite{mahimahi}              & --                           & base paper~\cite{bluebottle} \\
        \bottomrule
    \end{tabular*}
\end{table}

\section{Lean Formalization}
\label{app:lean}

Every statement of \Cref{app:full-proofs} is machine-checked in Lean~4, on the open-source formalization of \mysticeti-family rules,\footnote{\leanlink} with the base protocols' clauses as hypotheses (\Cref{app:lean-assumptions}). The development contains no \texttt{sorry} and uses no axioms beyond Lean's three standard ones: propositional extensionality, quotient soundness, and choice. Every assumption is a hypothesis of the theorem that needs it, so each result is either proved outright on the DAG model or proved under the hypotheses that \Cref{app:lean-assumptions} lists.

The development spans about $21$k lines, partitioned so that a human reader audits only definitions and theorem statements, about $7$k lines. The proofs, about $10$k lines, are AI-generated, checked by the kernel and need no reading. A further $3$k lines of executable witnesses evaluate every definition on concrete \pdag{}s before any theorem uses it, so that a definition admitting no instance fails the build instead of making the theorems above it hold for free. A script enforces the partition: statement files are proof-free, model files carry no theorems, and no synchrony hypothesis reaches the rule's own properties.

The witnesses include the \pdag{}s on which a claim fails: a fully connected \pdag{} on which a floor read from the anchor's wavelength derives both a direct commit and an anchored skip of one slot (\Cref{fig:anchor} draws the skip on a sparser \pdag{}); a rotating-leader family, at every horizon, on which the selector of \Cref{alg:update} keeps $\period = 4$ and no block above round $2$ is ever output while the \asyncslots commit on schedule, together with the recovery once the failover hands the period to $1$ and a run of the coin decides the stalled slot; a \pdag{} on which two hops of the anchor search are both led by the one Byzantine validator with probability at least $19/256$, against the $(b/n)^2 = 1/16$ that a per-hop bound would claim; and two views of one \pdag{} that, reading the \controlslots off the coin shares they hold rather than by rule, take their pivot at different rounds.

\subsection{The Model and the Rule}
\label{app:lean-model}

The model is the \pdag{} and a clock that stamps each block: blocks carry a round, an author and references, a view is a set of blocks closed under reference, and a schedule assigns each slot a round, a leader and a kind. There are no validators and no messages; \Cref{app:lean-assumptions} lists the hypotheses that stand in for them.

The wavelength function is a parameter. One map sends a kind to its wave, \wsync at the synchronous kind and \wasync at the asynchronous one, and a second sends a round to its kind, asynchronous at every $\period$-th round and synchronous elsewhere. For the $3f + 1$ pair, safety holds at any wavelength function of at least two rounds per kind, the paper's $\wlof{r}$ being one; the liveness results bound it as each pair needs. The $3f + 1$ pair's instances of \Cref{thm:liveness-ps} ask $\wl \ge 3$: at $\wl = 2$ the vote round $r + \wl - 2$ is the slot's own round, where no block but the candidate itself votes for it, so none is certified and nothing commits. BlueBottle's pair reads no certify round, so $\wsync = 2$ suffices for it.

For the $3f + 1$ pair, \sysname is one anchored rule whose data at a slot of kind $\kappa$ are \mahimahi's at wave $\wl(\kappa)$: the certificate-quorum direct commit, the slot blame as a direct skip, one certified link, and a wave offset of $\wl(\kappa) - 1$, so that an anchor sits at round $r + \wl(\kappa)$ or above. At a constant wavelength the rule \emph{is} \mahimahi's at that wave, by definition, and at the constant wave $3$ its verdicts are \mysticeti's by a proved equivalence (\Cref{thm:conservativity}).

The interface of \Cref{sec:interface} is formalized as a composite: given a family of anchored rules, one per kind, a slot of kind $\kappa$ takes its wave offset, direct predicates and links from the rule of that kind, while the link count and tie-break, which the relation reads without reference to a slot, come from the family. If every rule of the family satisfies the laws, which include \cref{clause:a2,clause:a3} and the uniqueness of \Cref{lem:cert-unique} in the relation's terms, and the family agrees on those two parameters, then the composite satisfies them, each law at a slot being that slot's own rule's and the anchor's rule never entering. \Cref{thm:agreement} is that statement, and both pairs the paper instantiates are instances of it: the $3f + 1$ one by definition, since its two rules are one rule read at two waves, and BlueBottle's from the two variants' own laws, which agree on the link count and the tie-break.

The period is a configuration-sequence model written for this paper. Rounds are grouped into intervals; the \controlslots of a scan are the multiples of the interval's period up to its boundary and the multiples of \pmaxperiod above it; the state a scan carries is the period, the agreed output's next slot and the round of its last committed leader; and the update rule is any function of the pivot block and the current period. The derivation is a relation, so the gating rule of \Cref{sec:adaptive-loop} holds by construction: a state for interval $j + 1$ exists only once the scan of interval $j$ has closed, with a pivot or with none. Waiting is the absence of a derivation, and a validator that has not closed a scan has no wavelength for the rounds above it.

What the arc reuses from the existing formalization is the \pdag{} core, the decision rule at a fixed wavelength, the coin modeled by its effect, and the counting lemma. What is new for this paper is the decision rule with a per-slot wavelength, the anchor floor at $r + \wlof{r}$, agreement across slot kinds at each validator's own derived schedule, the \controlverdicts as \rasync on the sub-schedule of one scan's \controlslots, the coin as a uniform distribution against an adversary that adapts to every earlier draw, the handover corollary as the one point where the two rules interact, the configuration-sequence model of the period with its agreement theorem for any deterministic update rule, and the liveness results, all but two stated for any pair of lawful rules, read through a clause per rule for the synchronous commit, the silent leader's skip and the counting floor.

\subsection{What the Model Assumes}
\label{app:lean-assumptions}

Because the model has no execution, five of the paper's assumptions enter as hypotheses rather than as derived facts. Three of them the base protocols already make, in their own formalizations as much as in their papers. The other two are new to the formalization: one states how the coin is drawn and what the \pdag{} may depend on, the other the round by which \cref{clause:a1}'s reference rule has spread a block.

\para{Inherited from the base protocols}
\emph{Block production and delivery} enter as populated waves, wherever a liveness result derives a commit from the \pdag{}: a claim asks that a quorum's blocks occupy the rounds its argument reads. \mahimahi's counting lemma takes the same hypothesis and no network hypothesis besides, so nothing is added to it here. \emph{Partial synchrony} enters as a condition on references rather than as a clock, namely that from some round on every block of the reliable quorum references every such block of the round below. That condition can be guaranteed only after \pgst and is what \mysticeti's post-\pgst liveness becomes when stated on references alone; the honest-leader commit of \Cref{thm:liveness-ps}, and hence \cref{clause:a4}, is proved under it. The development also checks that commit under two pacing disciplines, both with a timeout that comes to clear $2\Delta + \mathrm{proc}$ from \pgst on. The timed one establishes the condition. The reactive one does without it and asks instead for the waits of \Cref{sec:interleave-schedule}: the leader wait at the rounds that carry one and, at wave $3$, the vote wait, in which a certifier either references the votes or waits its timeout out. \emph{A fair leader schedule} is what \mysticeti's liveness rests on as well. The round count after \Cref{thm:liveness-ps} needs the sharper form $\wsync\,(n - |T|) < n$ for the reliable quorum $T$, because the hops of the floor chain interact: under it the chain reaches a reliably led landing within $n - |T|$ hops, within $b$ once the remaining faulty validators have crashed, and hence within $(b + 1)(\wsync + f)$ rounds. At a finite period the model checks the count too, under $\wsync\,(n - |T|) + \wsync\lceil n / \period \rceil < n$, with $n - 1$ rounds to an honestly led \syncslot in place of $f$ and every \asyncslot above the floor decided; with $|T| = n - f$ at $\wsync = 3$ and $n = 3f + 1$ no period meets that condition.

\para{New for this paper}
No base formalization states either of the two new hypotheses. \emph{The coin's unpredictability} (\cref{clause:a5}) enters as uniform draws, and the per-round floor of committed candidates that \cref{clause:a5} supplies is a hypothesis the model does not derive: the floor may depend only on the draws already revealed, and the adversary may rebuild the \pdag{} at every draw, so long as it does not shrink that floor with the round's own coin. \mahimahi's formalization models its coin by its effect alone, as a clause saying that the leader keeps landing on committed candidates, and leaves uniformity in prose; the bounds of \Cref{thm:liveness-async} and of the round counts after \Cref{thm:liveness-ps} are the first that need the floor written down, which is why \cref{clause:a5} states it. \emph{Validity under asynchrony} rests on the consequence of \cref{clause:a1}'s reference rule, that a block every honest validator holds by some round stays in every later honest block's causal history. The substrate's references sit one round back, so the rule and the round it yields lie outside the model, and the claim takes that round as its hypothesis.

Deriving block production, the reactive half of the pacing, and the reference rule from the assumptions themselves needs an execution model; the proofs of \Cref{app:full-proofs} take the same steps informally. Under these hypotheses the model does check the quantitative claims: the expected wait above an \asyncslot floor at $\period = 1$, as $1/p^{\wasync}$ blocks of \wasync coin rounds at both waves the pair accepts, and the expected $1/(1 - (1 - p)^{c})$ intervals before a scan finds its pivot at a given count $c$ of \controlslots.

The ``with probability $1$'' of \Cref{thm:liveness-async} needs one more step, since a finite record populates finitely many rounds and can carry only a tail that vanishes with the horizon. It is checked twice: as a tail over blocks of $\wasync \cdot \pmaxperiod$ coin rounds, one block opening every $\lceil \wasync \cdot \pmaxperiod / \pinterval \rceil$-th interval, and then almost surely over a sequence of records with the coin drawn as a process on the infinite product. A block is $\wasync \cdot \pmaxperiod$ rounds long because above an interval boundary a scan reads only the multiples of \pmaxperiod, so that is what it takes to hold \wasync consecutive \controlslots and settle the scans below. The tail asks of the interval only $\pmaxperiod \le \pinterval$ and $\wasync \le \pinterval$: at the implementation's $\pinterval = 128$ and $\pmaxperiod = 64$, a block opens every third interval at $\wasync = 5$ and every second at $\wasync = 4$.

\subsection{Results}
\label{app:lean-results}

\Cref{tab:lean} maps each statement of \Cref{app:full-proofs} to its machine-checked counterpart. Each Lean result is an audited statement file with a separate kernel-checked proof.

\begin{table}[tp]
    \centering
    \caption{Statements of \Cref{app:full-proofs} and their machine-checked counterparts.}
    \label{tab:lean}
    \scriptsize
    \begin{tabularx}{\textwidth}{@{}>{\raggedright\arraybackslash}p{1.85cm}@{\hspace{11pt}}>{\raggedright\arraybackslash}p{1.45cm}@{\hspace{11pt}}X@{}}
        \toprule
        Statement                     & Lean result               & Content                                                                                                                                                                                                                                                                                                                                                                                                                                                                                                                                                                \\
        \midrule
        \Cref{lem:cert-unique}        & Safety                    & A directly skipped slot has no certificate for any candidate, and two certified candidates of one author and round coincide, at the slot's own wave; for BlueBottle's pair, per variant, a candidate with $n - f$ supporters at its decision round leaves every rival below $n - 3f$, and a direct skip leaves every candidate below it.                                                                                                                                                                                                                               \\
        \Cref{lem:intersection}       & Safety                    & A directly committed candidate of kind $\kappa$ at round $r$ is certified in the causal history of every block at round $r + \wl(\kappa)$ or above, whatever wave that block's own slot carries; for BlueBottle's pair, per variant, $n - 3f$ of the candidate's supporters lie there.                                                                                                                                                                                                                                                                                 \\
        \Cref{thm:agreement}          & Safety,\newline Interface & Two views deciding one slot reach the same verdict by any routes; at the interface level, a family of rules whose laws hold and that agree on the link count and tie-break composes into one whose laws hold, and both pairs the paper instantiates are such a composite, the $3f+1$ one by definition and BlueBottle's from the two variants' own laws.                                                                                                                                                                                                               \\
        \Cref{cor:handover}           & Safety                    & A slot directly committed in one view is committed by every view that finds a committed anchor for it, whichever rule decides that anchor, and no view skips it.                                                                                                                                                                                                                                                                                                                                                                                                       \\
        \Cref{cor:order-integrity}    & Ledger                    & Over a prefix each view has settled, the committed-leader sequences and the ledgers coincide, and a ledger never drops a block and delivers each block at one slot.                                                                                                                                                                                                                                                                                                                                                                                                    \\
        \Cref{thm:liveness-ps}        & Liveness                  & The honest-leader commit, the crashed-leader skip from $n - f$ blames, partial dissemination not deferring an \asyncslot's decision, and the descent of the floor chain from a commit above; the round counts at $\period = \infty$ and $\period = 1$. Stated for any lawful rule except partial dissemination and the reactive pacing, which are stated per pair; checked at both pairs.                                                                                                                                                                              \\
        \Cref{thm:liveness-async}     & Period,\newline Coin      & Above every interval some interval finds a pivot almost surely, the \controlverdicts of a scan settling once \wasync consecutive \controlslots commit directly; the failover hands the next interval $\period = 1$; at $\period = 1$ a run of \wasync commits decides every slot below it; and the tail vanishes, almost surely and for every slot at once. Stated for any pair of lawful rules with $\wsync \le \wasync$ and the counting floor a parameter: $n - f - b$ for the $3f + 1$ pair at $\wasync = 5$, $1$ at $\wasync = 4$, and $n - 3f$ for BlueBottle's. \\
        \Cref{thm:agreement-period}   & Period                    & The \controlslots of a scan are a function of the period, the period bound and the boundary, their verdicts agree across views, and the state is agreed under any update rule and at each validator's own derived schedule, for any lawful rules and at both pairs.                                                                                                                                                                                                                                                                                                    \\
        \Cref{thm:conservativity}     & Safety                    & At a schedule of one kind the composite decides exactly as that kind's rule, for any family that agrees on the link count and tie-break, and so at both pairs; for the $3f + 1$ pair the rule at a constant wavelength is \mahimahi's at that wave by definition, and at the constant wave $3$ its verdicts are exactly \mysticeti's.                                                                                                                                                                                                                                  \\
        \Cref{lem:replay-starvation}  & Coin,\newline Replay      & $c_r / n$ is at least the counting fraction on the \pdag{} under any scheduling, $(n - f - b)/n$ for the $3f + 1$ pair at $\wasync = 5$ and $(n - 3f)/n$ for BlueBottle's, and at both pairs also on the window once a quorum, honest for BlueBottle's pair, has populated within it the rounds the counting lemma reads; a probe's success is a support quorum the \pdag{} holds.                                                                                                                                                                                     \\
        \Cref{cor:bab}                & Broadcast                 & \Cref{def:bab} clause by clause over settled prefixes: agreement, integrity and total order, validity after \pgst with the first commit two rounds up and, under asynchrony, with the first commit above the round by which the reliable validators have referenced the block.                                                                                                                                                                                                                                                                                         \\
        \Cref{alg:update}             & Replay                    & The window's evidence at both pairs, read through the rule's support, the three passes in exact rationals, the hysteresis and the tie rule, the odd canary probing every candidate, and every answer a divisor of \pmaxperiod within $[1, \pmaxperiod]$.                                                                                                                                                                                                                                                                                                               \\
        Parameters\newline and guards & Period,\newline Replay    & Gating rule; both bounds on \pinterval; the weaker one alone leaves a window resolving an \asyncslot at some pivots and none at others; the warm-up interval; every derived period in range.                                                                                                                                                                                                                                                                                                                                                                           \\
        \Cref{app:commit-prob}        & Timeout                   & A certificate forms with probability $1$ on a unanimous vote round and $f/n$ when one block abstained, and the $1/P_2$ wait is the layer cake of the tail probabilities.                                                                                                                                                                                                                                                                                                                                                                                               \\
        \bottomrule
    \end{tabularx}
\end{table}

\para{What is not checked}
For BlueBottle's pair everything is checked except the agreement of the \outputreading and the \controlreading on direct verdicts (\Cref{app:full-proofs}) and \Cref{app:commit-prob}, which concerns the certificate layer that pair lacks. The replay's window is read through each pair's support, the decision-round votes for BlueBottle's, and \Cref{lem:cert-unique,lem:intersection} are checked per rule, at the slot's wave for the $3f + 1$ pair and inside each variant's own development for BlueBottle's. Outside the model altogether are the evaluation of \Cref{sec:evaluation} and any reading of the schedule's clock as wall-clock time.

\section{Partially Synchronous Protocols Under a Mistimed Leader Timeout}
\label{app:commit-prob}

A commit rule that decides on one round of evidence survives asynchrony better than one that needs two. DAG-protocol designers know this as folklore; to the best of our knowledge it has never been written down. We formalize it here.
We consider the two partially synchronous protocols of this paper, \mysticeti and BlueBottle-PS, when their leader timeout $T$ is smaller than the link delay $D$. The commit rule of BlueBottle-PS ($\wsync = 2$) keeps committing with a constant probability. The commit rule of \mysticeti ($\wsync = 3$) stalls. Two observations explain the difference: a larger quorum increases the probability of referencing the leader block, and requiring a single layer of evidence increases the probability of a direct commit.
This is not a liveness statement: the delay $D$ is bounded, so configuring $T \ge D$ restores both rules (\Cref{sec:protocol}). Instead, we explain a latency collapse under a too-small timeout, and show why it depends on the number of quorum layers.

\subsection{Observation 1: Minimum-Quorum References}
\label{app:commit-prob-min-quorum}

A validator waits up to $T$ for the leader block before proposing, and the timer is armed at its own proposal. With every link delayed by $D > T$, the timer fires before any block of the round arrives, so the validator proposes the instant its threshold clock reaches $n - f$ blocks. As a consequence, every block carries minimum-quorum references. It includes exactly $n - f$ previous-round blocks, which are the earliest to arrive, and never more. We verified this on the DAG of the fixed-delay runs at $n = 10$: every block of both rules references exactly $n - f$ blocks.

We model this by assuming that at each validator the arrival order of a round's blocks is uniformly random. We assume equal stake and that a validator always references its own block.
A validator references its own block plus the $n-f-1$ earliest blocks of the other $n-1$. A leader block $L$ is therefore referenced with probability
\begin{equation}
    p = \frac{n-f-1}{n-1}.
    \label{eq:link}
\end{equation}
The chance of referencing the leader block is thus larger when the quorum threshold required to advance a round, $n - f$, is larger. At $n = 10$, BlueBottle-PS waits for $n - f = 9$ blocks and references the leader with probability $0.889$, while \mysticeti waits for $n - f = 2f+1 = 7$ blocks and gets $0.667$. The seemingly harder quorum helps, because the same quorum sets how long a validator waits before proposing.

\subsection{Observation 2: Layers of Evidence}
\label{app:commit-prob-layers}

\paragraph{One layer of evidence.}
The commit rule of BlueBottle-PS, at $n \ge 5f+1$, operates with $\wsync = 2$. The votes of the single decision round are the certificates (\Cref{sec:model}).
With $p$ from~\eqref{eq:link}, the total votes $V$ for $L$ follow a binomial distribution, with the leading $1$ being $L$'s author voting for itself:
\begin{equation}
    V = 1 + \mathrm{Bin}(n-1,\,p).
    \label{eq:votes}
\end{equation}
A direct commit requires at least $n-f$ votes for the leader block, so the commit probability is
\begin{equation}
    P_1 = \Pr[V \ge n-f] = \sum_{k = n-f-1}^{n-1} \binom{n-1}{k}\, p^{k} (1-p)^{n-1-k}.
    \label{eq:p1}
\end{equation}

\paragraph{Two layers of evidence.}
The commit rule of \mysticeti, at $n \ge 3f+1$, operates with $\wsync = 3$. A certificate is a certify-round block referencing $n - f$ votes, and a direct commit needs $n - f$ certificates.
With minimum-quorum references, a certifier holds exactly $n-f$ references. So all of them must be votes. With $v$ votes in the round, the chance of a certificate forming is hypergeometric:
\begin{equation}
    \pi(v) = \frac{\binom{v}{n-f}}{\binom{n}{n-f}}.
    \label{eq:cert}
\end{equation}
This probability is equal to $1$ at $v = n$ and $f/n$ at $v = n-1$, and is negligible below. At $n = 10$ the $v = n-1$ term contributes $0.001$ to the commit probability, hence
\begin{equation}
    P_2 \approx \Pr[V = n].
    \label{eq:p2}
\end{equation}
The second layer turns the rule into an all-or-nothing filter: the slot commits directly only if every vote-round block voted.

Under independent arrival orders, $\Pr[V = n] = p^{\,n-1}$. This yields $0.026$ at $n = 10$ and $4 \cdot 10^{-9}$ at $n = 50$. Independence is the one optimistic assumption. Arrival orders are positively correlated across validators because a block proposed early arrives early everywhere. On the fixed-delay DAG two validators' reference sets overlap in $3.7$ blocks against $3.1$ under independence, and $V$ has fatter tails with the same bulk. $P_1$ reads the bulk and is unaffected, with a measured $\Pr[V \ge n-f] = 0.652$ against $0.650$. $P_2$ reads the extreme tail, so $p^{\,n-1}$ is a floor. We measured $\Pr[V = n] = 0.10$ against $0.026$ at $n = 10$.

The consequence for progress follows. An undecided slot is decided only by the next direct commit above it, its anchor, and that commit decides every slot below it at once. Nothing accumulates; every slot simply waits for the next direct commit, which takes $1/P_2$ rounds in expectation. At $n = 50$ it never comes.

\subsection{Experimental Validation}
\label{app:commit-prob-validation}

\paragraph{Fixed delay.}
We tested a setup where every link is delayed by a fixed $D = 800\,$ms, with $T = 100\,$ms. We ran \mysticeti and BlueBottle-PS at $n \in \{10, 50\}$ across three seeds. Every run advances at $1.2$ rounds per second. The direct-commit rates per round are detailed in \Cref{tab:trunc}. At $n = 10$ the one-layer rule commits at the predicted rate. The two-layer rule commits directly in about one round in ten. This is four times the independent-order floor, but decisions arrive in bursts separated by $20$--$170\,$s in which no slot is decided. Transactions wait for the next burst, with a mean latency of $16$--$90\,$s, and one seed decides nothing in $180\,$s. At $n = 50$ the one-layer rule again matches $P_1$ within $3\%$. The two-layer rule decides nothing in any seed. The DAG advances at the same pace of about $210$ rounds and $1700$ leader timeouts, but the commit rule never fires, since without a direct commit there is no anchor either.

\paragraph{Partial asynchrony.}
We apply the random asynchronous model of Danezis et al.~\cite{random-network}, in which messages follow a random rather than an adversarial schedule, as instantiated in \Cref{sec:evaluation}: each message is delayed by $100$--$150\,$ms, past $T = 100\,$ms, with probability $\phi$. We evaluated $n = 10$ and $n = 50$ across three seeds. \Cref{fig:phi} plots latency and the direct-commit rate. Both rules interpolate between their healthy value and their minimum-quorum value of \Cref{tab:trunc}. BlueBottle-PS's direct-commit rate settles on the $P_1$ plateau of $0.74$ from $\phi \approx 0.3$. \mysticeti's rate falls toward its measured $\Pr[V = n]$ of $0.10$, which is above the independent-order floor. There is no threshold. \mysticeti's latency grows about $1.4$ times per five points of $\phi$. Sampling windows without any commit appear at $\phi = 0.25$ and dominate from $\phi = 0.40$. BlueBottle-PS pays at most five times its healthy latency and stays live throughout. The partial and full random-network panels of \Cref{sec:evaluation} are the $\phi = 0.3$ and $\phi = 1.0$ points. At $n = 50$ the same shape appears one step earlier: BlueBottle-PS saturates at $1.1$\,s with a direct-commit rate of $0.55$--$0.60$, on its closed form $P_1 = 0.59$, while \mysticeti's rate has fallen to $0.17$ by $\phi = 0.2$ and its latency risen to $2.7$\,s. Beyond that no slot is decided, and the runs cannot complete, as the simulator's cost per round grows with the undecided slots; \Cref{fig:phi} marks these points as off-scale. The closed forms describe only the minimum-quorum regime of \Cref{tab:trunc}, so they do not predict where the stall begins.

\begin{table}[!tp]
    \centering
    \caption{Direct commits per round and latency with minimum-quorum references ($D = 800\,$ms on every link, $T = 100\,$ms), measured over three seeds against the closed forms \eqref{eq:p1} and \eqref{eq:p2}.}
    \label{tab:trunc}
    \scriptsize
    \begin{tabular*}{\textwidth}{@{\extracolsep{\fill}}llccl@{}}
        \toprule
        Rule                         & $n$  & Direct commits per round & Model                         & Latency                  \\
        \midrule
        \mysticeti ($\wsync = 3$)    & $10$ & $0.10$ (one seed: none)  & $P_2 = 0.027$                 & $16$--$90\,$s, in bursts \\
        \mysticeti ($\wsync = 3$)    & $50$ & $0$ in all seeds         & $P_2 \approx 4 \cdot 10^{-9}$ & stalled, no decision     \\
        BlueBottle-PS ($\wsync = 2$) & $10$ & $0.68$--$0.75$           & $P_1 = 0.736$                 & $3.9$--$4.1\,$s          \\
        BlueBottle-PS ($\wsync = 2$) & $50$ & $0.57$--$0.62$           & $P_1 = 0.588$                 & $4.9$--$5.4\,$s          \\
        \bottomrule
    \end{tabular*}
\end{table}

\begin{figure}[tp]
    \centering
    \includegraphics[width=\linewidth]{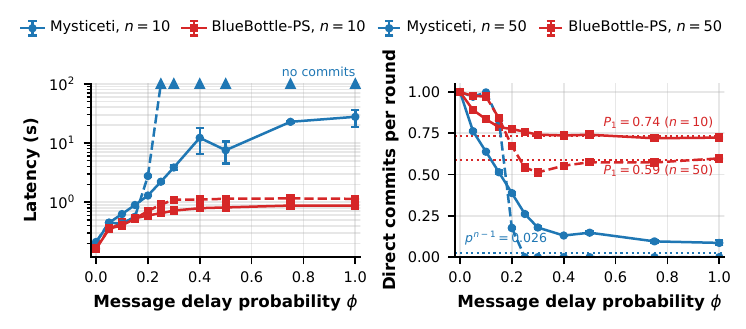}
    \caption{\mysticeti and BlueBottle-PS against the probability $\phi$ that a message is delayed past the leader timeout ($n = 10$ solid, $n = 50$ dashed, three seeds, $\phi = 0$ is the runs' healthy phase). Left: mean latency, log scale. Right: direct commits per round, with the closed forms $P_1$~\eqref{eq:p1} at both committee sizes and the independent-order floor $p^{\,n-1}$~\eqref{eq:p2} at $n = 10$ as guides. \mysticeti at $n = 50$ commits nothing beyond $\phi = 0.2$; those points are drawn off-scale at the top of the latency panel and at zero commits (see text).}
    \label{fig:phi}
\end{figure}

\section{Evaluation Details}
\label{app:evaluation-details}

We present the full setup behind \Cref{sec:evaluation}. We detail the network conditions and explain how to read the figures. We then give the per-panel numbers for both pairs of protocols at $n = 10$ and $n = 50$, and report the runs at $n = 10$.

\subsection{Setup}
\label{app:eval-setup}

\para{Parameters}
\Cref{tab:eval-parameters} lists the parameters for the simulator and the load generator in its top block. The quorum timeout defines the wait for stragglers past the quorum on rounds without a leader wait. The load generator batches transactions every $50$\,ms.
The bottom block of \Cref{tab:eval-parameters} lists the four parameters of \sysname. The \pinterval defines the replay window and the reaction horizon. This spans about $6$\,s at the healthy round rate. The \algvar{maxPeriod} is the starting period, setting one \asyncslot in $64$. The hysteresis is the threshold of improvement a candidate period must show to be adopted. The canary keeps the leader wait on every $31$st asynchronous round, so the committed DAG records the synchronous rule's evidence at any period; $31$ is coprime to the candidate periods, so the probes fall on synchronous slots of every candidate.

We configure the base protocols as shown in \Cref{tab:rules}. We run \mahimahi at $\wasync = 5$.

\begin{table}[tp]
    \centering
    \caption{Simulator, load, and \sysname parameters common to every run at $n = 10$.}
    \label{tab:eval-parameters}
    \scriptsize
    \begin{tabular*}{\textwidth}{@{\extracolsep{\fill}}ll@{}}
        \toprule
        Parameter & Value \\
        \midrule
        Committee size $n$ & $10$ ($f = 3$ for the $3f+1$ pair, $f = 1$ for the $5f+1$ pair) \\
        Topology, per-link latency & full mesh, uniform in $25$--$50$\,ms per message \\
        Leader timeout, quorum timeout & $100$\,ms, $40$\,ms \\
        Offered load, transaction size & $1000$\,tx/s system-wide ($100$\,tx/s per validator), $512$\,B \\
        Load generator batching & every $50$\,ms \\
        Run & $450$\,s: healthy $0$--$30$, condition $30$--$330$, healthy $330$--$450$ \\
        Sampling, seeds & $5$\,s windows, $7$ seeds \\
        Leader schedule & round-robin, one leader per round \\
        \midrule
        \pinterval & $128$ rounds \\
        $\algvar{maxPeriod}$ (initial period) & $64$ \\
        Hysteresis $\epsilon$ & $10\%$ \\
        Canary & every $31$st asynchronous round \\
        \bottomrule
    \end{tabular*}
\end{table}

\para{Conditions}
\Cref{tab:eval-conditions} details the six conditions of \Cref{sec:evaluation} with their parameters. Every model adds delay on top of the link latency and never drops a message.
The leader-delay adversary delays only the current round-robin leader's outbound blocks. This operates as a mute rather than an eclipse, as the targeted validator still receives blocks. The adversary tracks the round from the blocks it sees and is blind to the coin. The hidden leaders of \asyncslots are therefore never hit. Crashes are permanent. The last $f$ validators stop at $30$\,s.
The random-delay conditions instantiate the random asynchronous model of Danezis et al.~\cite{random-network}. Delays are drawn per message independently.

\para{Measurement}
We measure latency from a transaction's submission to the commit of its block, per replica, as the mean over each $5$\,s window. A panel plots the median across the seven seeds after a $3$-tick rolling median. This rolling median removes isolated single-window spikes of a backlogged protocol. Windows without a commit appear as gaps.
The bottom row of each figure breaks its axis when the partially synchronous protocol runs off the lower band.
The period trace is the per-replica period averaged across the committee. A value between two candidates indicates replicas mid-transition.

\begin{table}[tp]
    \centering
    \caption{The six network conditions of \Cref{fig:weather-mm,fig:weather-bb}, applied from $30$\,s to $330$\,s.}
    \label{tab:eval-conditions}
    \scriptsize
    \begin{tabularx}{\textwidth}{@{}l@{\hspace{8pt}}l@{\hspace{8pt}}>{\raggedright\arraybackslash}X@{}}
        \toprule
        Panel & Condition              & Parameters                                                                        \\
        \midrule
        (i)   & small leader delay     & leader's outbound blocks delayed by $30$\,ms, below the timeout                   \\
        (ii)  & large leader delay     & leader's outbound blocks delayed by $125$\,ms, above the timeout                  \\
        (iii) & permanent crash faults & the last $f$ validators stop at $30$\,s ($3$ / $1$), no recovery                  \\
        (iv)  & partial random network & each message delayed with probability $0.3$, uniformly in $100$--$150$\,ms        \\
        (v)   & full random network    & the same with probability $1$                                                     \\
        (vi)  & high jitter            & each message delayed by an exponential draw of mean $75$\,ms, capped at $400$\,ms \\
        \bottomrule
    \end{tabularx}
\end{table}

\subsection{Per-Panel Results}
\label{app:eval-results}

\Cref{tab:eval-results} provides the degraded plateau of the three lines per pair and panel. We calculate the plateau as the median of the $5$\,s windows between $70$\,s and $330$\,s. The table lists the time for \sysname to reach period $1$ after the onset and the time to return to $64$ after the condition lifts. The healthy rows represent the first $30$\,s. ``Stalled'' marks a protocol that commits in fewer than one window in ten. \mysticeti in panel (v) of \Cref{fig:weather-mm} (\Cref{sec:eval-degraded}) commits in $43\%$ of the windows at a $30$\,s median.

We detail two features of the period trace highlighted in \Cref{sec:evaluation}. In the crash panel (iii) of \Cref{fig:weather-mm}, the period alternates between $64$ and a value in $\{1, 2, 4, 8\}$ every second or third interval. This occurs identically in every seed because a probe landing on a crashed round-robin leader reads as a starved slot. Each dip costs a few windows of \mahimahi latency and causes the $7\%$ gap to \mysticeti. After recovery in panels (ii), (iv), and (v) of \Cref{fig:weather-mm} (\Cref{sec:eval-degraded}), the period drops to $2$ for one interval $20$--$70$\,s after the condition lifts, and one $5$\,s window commits at about $0.25$\,s. In panel (i) the period drops to $1$ or $2$ for one interval once or twice during the condition, each time costing one window at about $0.3$\,s.

\begin{table}[!htbp]
    \centering
    \caption{Per-panel results at $n = 10$: degraded-plateau latency (ms) of the partially synchronous protocol, the asynchronous protocol, and \sysname; \sysname's time to reach period $1$ after the onset and to return to $64$ after the condition lifts.}
    \label{tab:eval-results}
    \scriptsize
    \begin{tabular*}{\textwidth}{@{\extracolsep{\fill}}llrrrll@{}}
        \toprule
        Pair & Panel & Sync & Async & \sysname & To period $1$ & Back to $64$ \\
        \midrule
        $3f+1$ & healthy & $206$ & $287$ & $208$ & stays at $64$ & \\
        & (i) small leader delay & $303$ & $293$ & $306$ & 1--2 one-interval drops & \\
        & (ii) large leader delay & stalled & $334$ & $344$ & $10$--$20$\,s & $10$\,s \\
        & (iii) crash faults & $238$ & $334$ & $255$ & dips, see text & \\
        & (iv) partial random network & $3710$ & $685$ & $684$ & $15$--$20$\,s & $10$\,s \\
        & (v) full random network & $29\,800$ & $1121$ & $1122$ & $20$\,s & $5$--$15$\,s \\
        & (vi) high jitter & $3114$ & $1017$ & $1022$ & $15$--$25$\,s & $10$\,s \\
        \midrule
        $5f+1$ & healthy & $162$ & $208$ & $163$ & stays at $64$ & \\
        & (i) small leader delay & $236$ & $218$ & $238$ & a few one-interval drops & \\
        & (ii) large leader delay & stalled & $371$ & $374$ & $10$\,s & $10$\,s \\
        & (iii) crash faults & $169$ & $218$ & $198$ & dips, see text & \\
        & (iv) partial random network & $714$ & $672$ & $685$ & wanders & $0$--$10$\,s \\
        & (v) full random network & $872$ & $810$ & $806$ & wanders & $0$--$10$\,s \\
        & (vi) high jitter & $929$ & $908$ & $915$ & wanders & $0$--$10$\,s \\
        \bottomrule
    \end{tabular*}
\end{table}

\subsection{Committee of 10}
\label{app:eval-n10}

We repeat the six conditions with the same parameters at $n = 10$ ($f = 3$ for the $3f+1$ pair, $f = 1$ for the $5f+1$ pair) in \Cref{fig:weather-mm,fig:weather-bb}, on a longer timeline, which the smaller committee affords: $30$\,s healthy, the condition until $330$\,s, and healthy again until $450$\,s. \Cref{tab:eval-results} gives the per-panel numbers.

\Cref{claim:good,claim:async,claim:faults,claim:adaptive} hold at $n = 10$ as well. In the healthy phases \sysname stays within $1\%$ of the partially synchronous protocol ($208$ vs $206$\,ms and $163$ vs $162$\,ms) and $22$--$28\%$ below the asynchronous one. It reaches period $1$ within $10$--$25$\,s of the onset in every panel where a protocol stalls and returns to $64$ within $5$--$15$\,s of the condition lifting, and it tracks the better protocol within $3\%$ under the large leader delay ($344$ vs $334$\,ms) and within $1\%$ under the network-wide conditions: $684$ vs $685$\,ms under the partial random delay, $1122$ vs $1121$\,ms under the full one, and $1022$ vs $1017$\,ms under jitter. The crash dips of panel (iii) cost $7\%$ and $17\%$ for the two pairs, against $2\%$ at $n = 50$. Two things differ from $n = 50$. \mysticeti survives jitter at a $3.1$\,s plateau rather than $5.7$\,s, and BlueBottle-PS ties its asynchronous variant under the network-wide conditions instead of falling behind it, so \sysname sits between the two variants or just below both ($685$, $806$ and $915$\,ms) while its period wanders across the candidates, which costs nothing.

\begin{table}[tp]
    \centering
    \caption{Per-panel results at $n = 50$: degraded-plateau latency (ms) of the partially synchronous protocol, the asynchronous protocol, and \sysname; \sysname's time to reach period $1$ after the onset and to return to $64$ after the condition lifts.}
    \label{tab:eval-results-50}
    \scriptsize
    \begin{tabular*}{\textwidth}{@{\extracolsep{\fill}}llrrrll@{}}
        \toprule
        Pair & Panel & Sync & Async & \sysname & To period $1$ & Back to $64$ \\
        \midrule
        $3f+1$ & healthy & $213$ & $294$ & $215$ & stays at $64$ & \\
        & (i) small leader delay & $305$ & $295$ & $308$ & 0--1 one-interval drops & \\
        & (ii) large leader delay & stalled & $301$ & $313$ & $10$--$20$\,s & $10$\,s \\
        & (iii) crash faults & $356$ & $455$ & $363$ & dips, see \Cref{app:eval-results} & \\
        & (iv) partial random network & stalled & $731$ & $732$ & $15$\,s & $10$\,s \\
        & (v) full random network & stalled & $1167$ & $1168$ & $20$\,s & $10$\,s \\
        & (vi) high jitter & $5713$ & $1300$ & $1473$ & $20$--$45$\,s & $10$\,s \\
        \midrule
        $5f+1$ & healthy & $168$ & $212$ & $169$ & stays at $64$ & \\
        & (i) small leader delay & $238$ & $214$ & $238$ & stays at $64$ & \\
        & (ii) large leader delay & stalled & $220$ & $229$ & $10$\,s & $10$\,s \\
        & (iii) crash faults & $217$ & $268$ & $221$ & dips, see \Cref{app:eval-results} & \\
        & (iv) partial random network & $1041$ & $702$ & $705$ & $10$--$30$\,s & $5$--$15$\,s \\
        & (v) full random network & $1176$ & $835$ & $837$ & $15$--$45$\,s & $5$--$15$\,s \\
        & (vi) high jitter & $1023$ & $935$ & $944$ & wanders, $15$--$80$\,s & $0$--$10$\,s \\
        \bottomrule
    \end{tabular*}
\end{table}

\begin{figure}[!htbp]
    \centering
    \includegraphics[width=\linewidth]{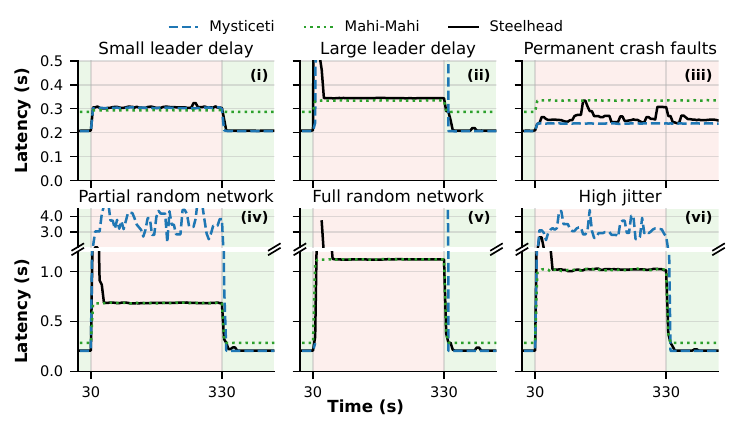}
    \caption{Latency over time of \mysticeti, \mahimahi, and adaptive \sysname ($3f+1$ pair, $n = 10$) under the six network conditions of \Cref{sec:eval-setup}, each applied from $30$\,s to $330$\,s (shaded).}
    \label{fig:weather-mm}
\end{figure}

\begin{figure}[!htbp]
    \centering
    \includegraphics[width=\linewidth]{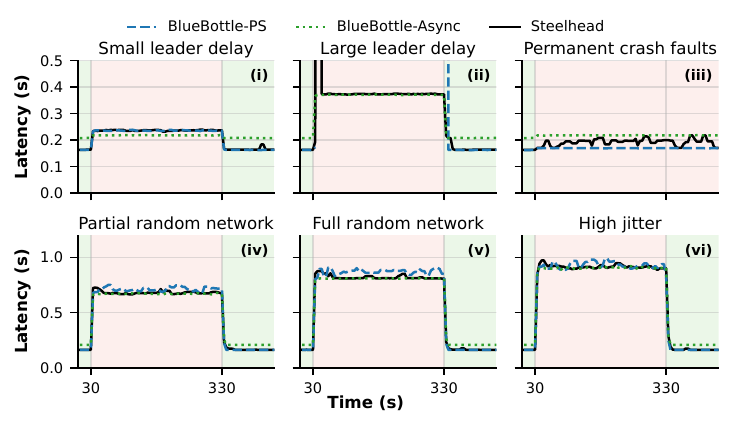}
    \caption{Latency over time of BlueBottle-PS, BlueBottle-Async, and adaptive \sysname ($5f+1$ pair, $n = 10$) under the six network conditions of \Cref{sec:eval-setup}, each applied from $30$\,s to $330$\,s (shaded).}
    \label{fig:weather-bb}
\end{figure}

\end{document}